\RequirePackage{fix-cm}
\documentclass{svjour3-HAL}                     
\usepackage{verbatim}
\usepackage{stackengine}
\usepackage{amssymb}
\usepackage{amsmath}
\usepackage{graphicx}
\usepackage{comment}
\usepackage[noend]{algorithmic}

\usepackage{amsmath}
\usepackage{wasysym}
\usepackage{booktabs}

\usepackage[ruled]{algorithm}
\usepackage{algorithmic}

\usepackage{url}
\urldef{\mailsa}\path|dominique.laurent@u-cergy.fr|
\urldef{\mailsb}\path|nicolas.spyratos@lri.fr|

\usepackage{todonotes}

\begin{document}

\title{Four-Valued Semantics for Deductive Databases}

\titlerunning{Four-Valued Semantics for Deductive Databases}

%
%
\author{Dominique Laurent\and Nicolas Spyratos}

\institute{Dominique Laurent \at ETIS Laboratory - ENSEA, CY Cergy Paris University, CNRS\\F-95000 Cergy-Pontoise, France\\ \mailsa 
\and
Nicolas Spyratos \at LISN Laboratory - University Paris-Saclay, CNRS\\
F-91405 Orsay, France\\
\mailsb\\~\\
{\bf Acknowledgment:} Work conducted while the second author was visiting at FORTH Institute of Computer Science, Crete, Greece (https://www.ics.forth.gr/)}

\date{}
\maketitle

\begin{abstract}
In this paper, we introduce a novel approach to deductive databases meant to take into account  the needs of  current applications in the area of data integration. To this end, we extend the formalism of standard deductive databases  to the context of Four-valued logic so as to account for {\em unknown}, {\em inconsistent}, {\em true} or {\em false} information under the open world assumption. In our approach, a database  is a pair $(E,R)$ where $E$ is the extension and $R$ the set of rules. The extension is a set of pairs of the form $\langle \varphi, {\tt v}\rangle$ where $\varphi$ is  a fact and {\tt v} is a value that can be true, inconsistent or false - but not unknown (that is, unknown facts are not stored in the database). The rules follow the form of standard Datalog$^{neg}$ rules but, contrary to standard rules, their head may be a negative atom.
\\
Our main contributions are as follows: $(i)$ we give an expression of  first-degree entailment in terms of other connectors and exhibit a functionally complete set of basic connectors not involving  first-degree entailment, 
$(ii)$ we define a new operator for handling our new type of rules and show that this operator is monotonic and continuous, thus providing an effective way for defining and computing database semantics, and $(iii)$ we argue that our framework allows for the definition of a new type of updates that  can be used in most standard data integration applications.
\end{abstract}

\section{Introduction}
In this paper, we present a novel approach meant to take into account  the needs of many current applications, specifically in the domain of data integration. Our purpose is to extend the concept of deductive databases \cite{CeriGT90,Ullman} to the context of Four-valued logic \cite{Belnap}, a  formalism known to be suitable for data integration, as it allows  to deal with {\em unknown}, {\em inconsistent}, {\em true} or {\em false}  information. We begin by illustrating our approach through an example used as our running example throughout the paper.

\medskip\noindent{\bf Running Example.}
Our example concerns the storage of bags of rice grains, considering two important factors that (among others) influence the design and development of optimum storage, namely color and humidity of the rice grains \cite{Batay}. 

We assume that each bag is tested for the color and humidity of its rice grains in two different sites, first just before leaving the rice farm and then just before entering the warehouse. The outcomes of these tests can be: humid or not humid (with respect to a humidity threshold); and white or not white (with respect to a color threshold). Based on these outputs, the following actions are taken:
\begin{itemize}
\item If the grains are not humid and white then store the bags in the warehouse.
\item If the grains are humid then do not store the bags but cure the grains.
\item If the grains are not white then do not store the bags but analyze further.
\end{itemize}
We assume that the tests are conducted by sensors: two sensors at the rice farm, one for humidity, denoted $H_1$, and one for color denoted $W_1$; and two sensors at the warehouse denoted $H_2$ and $W_2$.  We also assume that, during a test, if the sensor is functioning then it returns a Boolean value (true or false), otherwise it returns no value. Under these assumptions, one of the following cases can appear for the sensors testing humidity (and similarly for the sensors testing color):
\begin{enumerate}
\item  The two sensors return the same value.
\item The two sensors return {\em different} values.
\item Only one of the two sensors returns a value.
\item Neither of the two sensors returns a value.
\end{enumerate}
In this setting, let $Humid(ID)$, denote the humidity state or `value' of a bag with identifier $ID$. Then the question is: what value should we assign to $Humid(ID)$ in each of the four cases above? In our formalism, we answer this question by `integrating' the outputs of $H_1$ and $H_2$ as follows (and similarly for the outputs of $W_1$ and $W_2$):
\begin{enumerate}
\item $Humid(ID)$ is set to the common value returned by the sensors.
\item $Humid(ID)$ is set to {\em inconsistent}, to mean that the sensors returned different values.
\item $Humid(ID)$ is set to the value returned by the sensor which returned a value.
\item $Humid(ID)$ is set to {\em unknown}, to mean that neither of the two sensors returned a value.
\end{enumerate}
As our example shows, we clearly need more than the standard truth values {\tt True} and {\tt False}, to express the cases 2 and 4 above. It will be seen that the Four-valued logic introduced in \cite{Belnap} provides the right formalism as it provides the additional truth values needed and also appropriate connectors to work with these additional truth values. For instance, using a connector denoted by $\oplus$ we can express all four cases above in a single expression: $Humid(ID)= H_1(ID) \oplus H_2(ID)$. 

The database is a pair $(E,R)$ where $E$ collects the sensor outputs and where $R$ is a set of rules describing how to integrate these outputs and how to treat the bags based on the integrated values. Formally, the elements of $E$ are pairs of the form $\langle \varphi, {\tt v}\rangle$ to represent the output of one sensor about a bag recognized by its identifier. In such pair $\varphi$ is a fact regarding the humidity or the color of a bag and ${\tt v}$ is its associated truth value. The rules expressing the integration of the sensor outputs and the conditions regarding the storage of the bags are as follows:

\smallskip
\begin{tabular}{ll}
$\rho_1: Humid(x) \leftarrow H_1(x) \oplus H_2(x)$&$\rho_5: Cure(x)  \leftarrow Humid(x)$\\
$\rho_2: White(x) \leftarrow W_1(x) \oplus W_2(x)$&$\rho_6: \neg Store(x)  \leftarrow \neg White(x)$\\
$\rho_3: Store(x)  \leftarrow \neg Humid,(x) \wedge White(x)$&$\rho_7: New\_test(x)  \leftarrow \neg White(x)$\\
$\rho_4: \neg Store(x)  \leftarrow Humid(x)$&
\end{tabular}

\smallskip\noindent
Although the rules above roughly look like standard Datalog rules with negation, the following basic differences have to be noticed:
\begin{enumerate}
\item
The body of a rule is not restricted to be a conjunction of literals; in fact we allow all available connectors to occur in the body of a rule.
\item
The head of a rule is not restricted to be an atom: negative literals are allowed, at the cost of generating contradictory facts.
\item
Contradictions are allowed in database semantics and treated as such, in the context of the Four-valued semantics introduced in \cite{Belnap}.
\end{enumerate}
To illustrate how our approach deals with such rules, we first give a rough overview of the basic notions used in our approach. First, in Four-valued logic, four truth values are considered, namely {\tt t}, {\tt b}, {\tt n} and {\tt f}, standing respectively for {\em true}, {\em inconsistent}, {\em unknown}{\footnote{The intuition explaining the notation {\tt b} and {\tt n} will be clarified later in this paper.}} and {\em false}.

In this context the pieces of information to be stored in the database extension are pairs of the form $\langle \varphi, {\tt v}\rangle$ where $\varphi$ is a fact (i.e. an atom with no variable) and ${\tt v}$ is one of the four truth values just mentioned. By such a pair, which we call {\em valuated pair} or v-pair for short, we mean that {\em `$\varphi$ has truth value ${\tt v}$'}. Moreover, we make the intuitively appealing convention that unknown facts are not stored, meaning that the database extension can {\em not} contain a v-pair of the form $\langle \varphi, {\tt n}\rangle$. We emphasize that, contrary to most database approaches in which only {\em true} pieces of information are stored, our approach allows to store {\em true}, {\em false} or even {\em inconsistent} pieces of information.

Continuing with our example, assume there are three rice bags with identifiers $101$, $202$ and $303$ for which the following sensor outputs and corresponding v-pairs are  stored in the database:
\begin{description}
\item{\em Regarding bag $101$:}
$H_1$ and $H_2$ both return  {\tt False}; this results in storing the two v-pairs $\langle H_1(101), {\tt f}\rangle$ and $\langle H_2(101), {\tt f}\rangle$  in the database extension.
$W_1$ returns {\tt True} but $W_2$ returns no value;  this results in storing the pair $\langle W_1(101), {\tt t}\rangle$ in the database extension.
\item{\em Regarding bag $202$:}
$H_2$ returns {\tt True} and $H_1$ returns no value; this results in storing the v-pair $\langle H_2(202), {\tt t}\rangle$ in the database extension.
$W_1$ returns {\tt False} while $W_2$ returns  {\tt true};  this results in storing the two pairs $\langle W_1(202), {\tt f}\rangle$ and $\langle W_2(202), {\tt t}\rangle$ in the database extension.
\item{\em Regarding bag $303$:}
$H_1$ and $H_2$ both return no value, $W_1$ returns {\tt False} and $W_2$ returns no value;  this results in storing the pair $\langle W_1(303), {\tt f}\rangle$ in the database extension.
\end{description}
Roughly speaking,  given a set $S$ of v-pairs, applying a rule $\rho$ is achieved as follows: for every instantiation of $\rho$ denoted $inst(\rho)$, the truth value of the body of $inst(\rho)$ is computed against $S$, and if this truth value is {\tt t} or {\tt b} then this truth value is assigned to the head of the $inst(\rho)$. Moreover, as more than one rule head may involve the same fact, in case of conflicting assignment, we apply the integration statements as done for the sensors. We illustrate this processing below.
\begin{enumerate}
\item
At the first step, the only rules that apply are $\rho_1$ and $\rho_2$.
\begin{itemize}
\item
Based on the v-pairs $\langle H_1(101), {\tt f}\rangle$ and $\langle H_2(101), {\tt f}\rangle$, $\rho_1$ generates the v-pair $\langle Humid(101), {\tt f}\rangle$ stating that the grains in bag $101$ are not humid.\\
As for identifier $202$, since the output of $H_1$ is missing, we consider the (non-stored) v-pair $\langle H_1(202), {\tt n}\rangle$, which combined by $\oplus$ with the stored v-pair $\langle H_2(202), {\tt t}\rangle$ generates $\langle Humid(202), {\tt t}\rangle$ stating that the grains in bag $202$ are  humid.\\
As for identifier $303$, since both $H_1$ and $H_2$ no value, $\rho_1$ generates no v-pair involving $Humid(303)$, meaning that the humidity of the grains in the bag $303$ is unknown. 
\item
As for $White(101)$, since $W_2$ returns no value, $\rho_2$ generates  the v-pair $\langle White(101), {\tt t}\rangle$ stating that the grains in bag $101$ are white.\\
As for $White(202)$, we notice that $W_1$ and $W_2$ disagree. In this case, $\rho_2$ generates  the v-pair $\langle White(202), {\tt b}\rangle$, meaning that the fact $White(202)$ is inconsistent, thus that the color of the grains in bag $202$ cannot be decided.\\
As for $White(303)$, since $W_2$ returns no value, $\rho_2$ generates the v-pair $\langle White(303), {\tt f}\rangle$, meaning that  the grains in bag $303$ cannot be considered white.
\end{itemize}
\item
The next step is based on the v-pairs earlier generated, namely: $\langle Humid(101), {\tt f}\rangle$, $\langle Humid(202), {\tt t}\rangle$, $\langle White(101), {\tt t}\rangle$, $\langle White(202), {\tt b}\rangle$ and $\langle White(303), {\tt f}\rangle$. The rules $\rho_3 \ldots \rho_7$ apply as follows:
\begin{itemize}
\item
Based on $\langle Humid(101), {\tt f}\rangle$ and $\langle White(101), {\tt t}\rangle$, $\rho_3$ generates the v-pair $\langle Store(101), {\tt t}\rangle$.
Considering $\langle Humid(202), {\tt t}\rangle$ and $\langle White(202), {\tt b}\rangle$, since the conjunction of the body is false, $\rho_3$ does not apply.
Since $Humid(303)$ is unknown and $White(303)$ is false, the conjunction of the body is false, entailing that $\rho_3$ does not apply.
\item
Since $Humid(101)$ is not true, $\rho_4$ does not apply.
Since $Humid(202)$ is true, $\rho_4$ generates $\langle Store(202), {\tt f}\rangle$.
Since $Humid(303)$ is unknown, $\rho_4$ does not apply.
\item
As above, since $Humid(101)$ is not true, $\rho_5$ does not apply, but $\rho_5$ generates  $\langle Cure(202), {\tt t}\rangle$ because $Humid(202)$ is true.
\item
Similarly, since $White(101)$ is not false, $\rho_6$ and $\rho_7$ do not apply.
Since $White(202)$ is {\em inconsistent}, $\rho_6$ and $\rho_7$ generate respectively $\langle Store(202), {\tt b}\rangle$ and $\langle New\_test(202), {\tt b}\rangle$.
Moreover, since $White(303)$ is false, $\rho_6$ and $\rho_7$ generate respectively $\langle Store(303), {\tt f}\rangle$ and $\langle New\_test(303), {\tt t}\rangle$.
\end{itemize}
After applying the rules, conflicting v-pairs involving $Store(202)$ appear, because $Store(202)$  has been found {\em false} by $\rho_4$ and {\em inconsistent} by $\rho_6$. In this case, we integrate these different truth values in much the same way as we did for the sensor outputs, stating that $Store(202)$ should be inconsistent. Therefore,  the v-pair $\langle Store(202), {\tt f}\rangle$ is removed from the result of this step.
\item
As no further v-pair can be generated by the rules based on the v-pairs generated in the previous steps, the processing stops and returns the set of all these v-pairs, which added to the database extension constitutes what we call the {\em database semantics}.
\end{enumerate}
The obtained database semantics is therefore the set of the following v-pairs:

\smallskip
$\langle H_1(101), {\tt f}\rangle$, $\langle H_2(101), {\tt f}\rangle$, $\langle H_2(202), {\tt t}\rangle$,

$\langle W_1(101), {\tt t}\rangle$, $\langle W_1(202), {\tt f}\rangle$, $\langle W_2(202), {\tt t}\rangle$, $\langle W_1(303), {\tt f}\rangle$,

$\langle Humid(101), {\tt f}\rangle$, $\langle Humid(202), {\tt t}\rangle$,

$\langle White(101), {\tt t}\rangle$, $\langle White(202), {\tt b}\rangle$, $\langle White(303), {\tt f}\rangle$,

$\langle Store(101), {\tt t}\rangle$, $\langle Store(202), {\tt b}\rangle$, $\langle Store(303), {\tt f}\rangle$,

$\langle Cure(202), {\tt t}\rangle$, $\langle New\_test(202), {\tt b}\rangle$, $\langle New\_test(303), {\tt t}\rangle$.

\smallskip\noindent
It is shown in this paper that the computation just described in an informal way is sound and its relationship with other related approaches is investigated. Moreover, some basic properties of the underlying Four-valued logic are stated, and among them this example raises the following question: could the rules $\rho_4$ and $\rho_6$ be replaced by the single rule $\rho_{46}: \neg Store(x) \leftarrow Humid(x) \vee \neg White(x)$? Whereas this question is answered positively in standard approaches to Datalog databases (\cite{CeriGT90,Ullman}) and in the Four-valued approach of \cite{Fitting91}, we argue that this replacement raises some issues.
\hfill$\Box$

\medskip\noindent
This work  is an extension of that in  \cite{Lau2019} where rule bodies are restricted to be conjunctions. The main contributions of this paper are as follows:
\begin{enumerate}
\item
We show that FDE ({\em First Degree Entailment}) implication, one of the standard implications in Four-valued logic, can be expressed in terms of the usual connectors.
\item
We exhibit a functionally complete set of basic connectors {\em not} involving FDE implication, contrary to the results in \cite{Arieli1998}.
\item
We generalize the rules by allowing negative literals in their heads and connectors other than negation, conjunction and disjunction in their bodies.
\item
We define a new immediate consequence operator for handling such rules, and we show that this operator is monotonic and continuous,  thus providing an effective way for defining and computing database semantics.
\item
We argue that our context allows for the definition of a new type of updates that  can be used in  data integration applications. Notice that to the best of our knowledge, the problem of database updating in a Four-valued logic framework has never been addressed in the literature.
\end{enumerate}
The paper is organized as follows: In Section~\ref{sec:background} we review the formalism related to Four-valued logic and we address the first two issues mentioned above. Section~\ref{sec:database-sem} is devoted to  the definitions of the syntax and the semantics of databases in the context of Four-valued logic. In Section~\ref{sec:updates}, we define two types of updates, one standard and another one related to data integration. Then, in Section~\ref{sec:rel-work} we review some of the approaches related to our work that can be found in the literature. 
Section~\ref{sec:conclusion} provides an overview of our approach and suggests research issues that we are currently investigating or that we intend to investigate in the next future.
\section{Background: Four-Valued Logic}\label{sec:background}
\subsection{Basics of Four-Valued Logic}
Four-valued logic was introduced by Belnap in \cite{Belnap}, who argued that this formalism could be of interest when integrating data  from various data sources. To this end,  denoting by {\tt t}, {\tt b}, {\tt n} and {\tt f} the four truth values, the usual connectives $\neg$, $\vee$ and $\wedge$ have been defined as shown in Figure~\ref{fig:truth-tables-con}. An important feature of this Four-valued logic is that it allows to compare truth values according to two partial orderings, known as {\em truth ordering} and {\em knowledge ordering}, respectively denoted by $\preceq_t$ and $\preceq_k$ and defined by:

\smallskip\centerline{
${\tt n}\preceq_k {\tt t}\preceq_k {\tt b}$~;~${\tt n}\preceq_k {\tt f} \preceq_k {\tt b}$ \qquad and \qquad 
${\tt f}\preceq_t {\tt n}\preceq_t {\tt t}$~;~${\tt f}\preceq_t {\tt b} \preceq_t {\tt t}$.}

\smallskip\noindent
To explain the choice of {\tt b} and {\tt n} as notation for {\em inconsistent} and {\em unknown}, let ${\cal V}=\{{\tt True},{\tt False}\}$ be the set of the usual  truth values. The four truth values in Four-valued logic can then be thought of as corresponding to the elements in the power set of ${\cal V}$, by associating respectively $\emptyset$, $\{{\tt False}\}$, $\{{\tt True}\}$, $\{{\tt True},\, {\tt False}\}$ with ${\tt n}$, ${\tt f}$, ${\tt t}$, ${\tt b}$. Then the notation ${\tt n}$ and  ${\tt b}$ can be read respectively as {\em none} and {\em both}. Notice also that, under this association, the ordering $\preceq_k$, the connectors $\oplus$ and $\otimes$ are respectively nothing but the restriction to the power set of ${\cal V}$ of set theoretic inclusion, union and intersection.

As in standard two-valued logic, conjunction (respectively disjunction) corresponds to minimum (respectively maximum) truth value, when considering the  truth ordering. It has also been shown in \cite{Belnap,Fitting91} that the set $\{{\tt t}, {\tt b}, {\tt n}, {\tt f}\}$ equipped with these two orderings has a distributive bi-lattice structure, where the minimum and maximum with respect to $\preceq_k$ are denoted by $\otimes$ and $\oplus$, respectively.

Not surprisingly, it should be emphasized that in this Four-valued logic some basic properties holding in standard logic do not hold. For example, Figure~\ref{fig:truth-tables-con} shows that formulas of the form $\Phi \vee \neg\Phi$ are not always true, independently from the truth value of $\Phi$. More importantly, it has been argued in \cite{Arieli1998,Hazen17,Tsoukias} that defining the implication  $\Phi_1 \Rightarrow \Phi_2$ by $\neg \Phi_1 \vee \Phi_2$, is problematic.

To see this,  we consider as in \cite{Belnap,Arieli1998,Hazen17,Tsoukias}, that {\tt t} and {\tt b} are the two {\em designated truth values}, because as mentioned above, these truth values are the only ones corresponding to sets containing {\tt True}. As a consequence, a formula $\Phi$ is  said to be {\em valid} if its truth value is designated, {\em i.e.,} either  {\tt t} or {\tt b}. 

As argued in \cite{Arieli1998,Hazen17,Tsoukias}, $\Rightarrow$ does not satisfy the deduction theorem, because the formula $\Phi$ defined by $(\Phi_1 \wedge (\Phi_1 \Rightarrow \Phi_2)) \Rightarrow \Phi_2$ is {\em not valid  for every truth value assignment}. Indeed based on Figure~\ref{fig:truth-tables-imp}, for every assignment $v$ such that  $v(\Phi_1)={\tt n}$ and $v(\Phi_2)={\tt f}$, we have $v(\Phi_1 \Rightarrow \Phi_2)={\tt n}$ and thus, $v(\Phi)={\tt n}$. As a consequence, we discard $\Rightarrow$ as the implication providing semantics to our rules.


\smallskip
Among the various implications  introduced in the literature, {\em First Degree Entailment} implication, or FDE implication,  denoted hereafter by $\to$ (\cite{Arieli1998,Hazen17}) is the most popular. We also mention another implication introduced in \cite{Tsoukias} and  denoted hereafter by $\hookrightarrow$. Each of these implications is associated with another implication, denoted by $\stackrel{*}{\to}$ and $\stackrel{*}{\hookrightarrow}$ whose role is explained next. The truth tables of all these implications are shown in Figure~\ref{fig:truth-tables-imp}.

Recall from  \cite{Arieli1998} (Corollary 9) that $\to$, is defined `from scratch' in the sense that it cannot be expressed using the other standard connectives $\neg$, $\vee$ and $\wedge$. As we shall see shortly we can provide an expression of $\to$ involving standard connectors in the formalism of \cite{Tsoukias}. It is also important to notice that as shown in \cite{Tsoukias}, $\Phi_1 \hookrightarrow \Phi_2$ is defined by $\sim \Phi_1 \vee \Phi_2$, where $\sim$ is a complement operator whose truth table is shown in Figure~\ref{fig:truth-tables-con}.

Moreover, since $\Phi_1 \to \Phi_2$ and $\neg\Phi_2 \to \neg\Phi_1$ are not equivalent, the implication $\Phi_1 \stackrel{*}{\to} \Phi_2$ is introduced in  \cite{Arieli1998,Hazen17} as a shorthand for $(\Phi_1 \to \Phi_2) \wedge (\neg\Phi_2 \to \neg\Phi_1)$. As a similar situation holds regarding $\hookrightarrow$,
$\Phi_1 \stackrel{*}{\hookrightarrow} \Phi_2$ is defined in \cite{Tsoukias} as  $(\Phi_1 \hookrightarrow \Phi_2) \wedge (\neg\Phi_2 \hookrightarrow \neg\Phi_1)$. 

In an attempt to compare these implications, we notice that, contrary to $\Rightarrow$, the formula $\Phi$ defined by $(\Phi_1 \wedge (\Phi_1 \leadsto \Phi_2)) \leadsto \Phi_2$ is valid when replacing $\leadsto$ with one of the implications $\to$, $\hookrightarrow$, $\stackrel{*}{\to}$ or $\stackrel{*}{\hookrightarrow}$.
%
%
%
%
It is also interesting to see that when merging the truth values {\tt t} and {\tt b} (respectively {\tt f} and {\tt n}) into a single value, say {\tt TRUE} (respectively {\tt FALSE}), the corresponding truth tables of $\to$ and $\hookrightarrow$ are that of the standard implication, while this is not the case for $\Rightarrow$, $\stackrel{*}{\to}$ and $\stackrel{*}{\hookrightarrow}$. This explains why we discard these three implications. However, the choice between $\to$ and $\hookrightarrow$ is not easy for the following reasons:
\begin{itemize}
\item
In \cite{Arieli1998,Hazen17}, it is argued that, similarly to two-valued implication, $\to$ satisfies the property that $v(\Phi_1 \to \Phi_2) = v(\Phi_2)$ whenever $v(\Phi_1)$ is designated. However, $\to$ does not satisfy the properties of $\hookrightarrow$ given below.
\item
Although  $\hookrightarrow$ does not satisfy the above property, it is argued  in \cite{Tsoukias} that, similarly to two-valued implication, $\hookrightarrow$ satisfies the property that $v(\Phi_1) \preceq_t v(\Phi_2)$ if and only if $v(\Phi_1 \hookrightarrow \Phi_2) = {\tt t}$. 
\end{itemize}
We draw attention on that none of these two implications satisfies all intuitively appealing properties that standard two-valued implication satisfies, among which contraposition is an  example.
%

%
\begin{figure}[t]
\begin{center}
{\footnotesize
\begin{tabular}{c|c}
\,$\varphi$\,&\,$\neg \varphi$\,\\
\hline
${\tt t}$&${\tt f}$\\
${\tt b}$&${\tt b}$\\
${\tt n}$&${\tt n}$\\
${\tt f}$&${\tt t}$\\
\end{tabular}
\qquad
\begin{tabular}{c|c}
\,$\varphi$\,&\,$\not\sim \varphi$\,\\
\hline
{\tt t}&{\tt b}\\
{\tt b}&{\tt t}\\
{\tt n}&{\tt f}\\
{\tt f}&{\tt n}\\
\end{tabular}
\qquad
\begin{tabular}{c|c}
\,$\varphi$\,&\,$\sim \varphi$\,\\
\hline
{\tt t}&{\tt f}\\
{\tt b}&{\tt n}\\
{\tt n}&{\tt b}\\
{\tt f}&{\tt t}\\
\end{tabular}

\vspace{.5cm}
\begin{tabular}{c|cccc}
\,$\vee$\,&{\tt t}&{\tt b}&{\tt n}&{\tt f}\\
\hline
{\tt t}&{\tt t}&{\tt t}&{\tt t}&{\tt t}\\
{\tt b}&{\tt t}&{\tt b}&{\tt t}&{\tt b}\\
{\tt n}&{\tt t}&{\tt t}&{\tt n}&{\tt n}\\
{\tt f}&{\tt t}&{\tt b}&{\tt n}&{\tt f}\\
\end{tabular}
\qquad
\begin{tabular}{c|cccc}
\,$\wedge$\,&{\tt t}&{\tt b}&{\tt n}&{\tt f}\\
\hline
{\tt t}&{\tt t}&{\tt b}&{\tt n}&{\tt f}\\
{\tt b}&{\tt b}&{\tt b}&{\tt f}&{\tt f}\\
{\tt n}&{\tt n}&{\tt f}&{\tt n}&{\tt f}\\
{\tt f}&{\tt f}&{\tt f}&{\tt f}&{\tt f}\\
\end{tabular}

\vspace{.5cm}
\begin{tabular}{c|cccc}
\,$\oplus$\,&{\tt t}&{\tt b}&{\tt n}&{\tt f}\\
\hline
{\tt t}&{\tt t}&{\tt b}&{\tt t}&{\tt b}\\
{\tt b}&{\tt b}&{\tt b}&{\tt b}&{\tt b}\\
{\tt n}&{\tt t}&{\tt b}&{\tt n}&{\tt f}\\
{\tt f}&{\tt b}&{\tt b}&{\tt f}&{\tt f}\\
\end{tabular}
\qquad
\begin{tabular}{c|cccc}
\,$\otimes$\,&{\tt t}&{\tt b}&{\tt n}&{\tt f}\\
\hline
{\tt t}&{\tt t}&{\tt t}&{\tt n}&{\tt n}\\
{\tt b}&{\tt t}&{\tt b}&{\tt n}&{\tt f}\\
{\tt n}&{\tt n}&{\tt n}&{\tt n}&{\tt n}\\
{\tt f}&{\tt n}&{\tt f}&{\tt n}&{\tt f}\\
\end{tabular}
}
\end{center}
\caption{Truth tables of basic connectors}
\label{fig:truth-tables-con}
\end{figure}

\begin{figure}[t]
\begin{center}
{\footnotesize
\begin{tabular}{c|cccc}
\,$\Rightarrow$\,&{\tt t}&{\tt b}&{\tt n}&{\tt f}\\
\hline
{\tt t}&{\tt t}&{\tt b}&{\tt n}&{\tt f}\\
{\tt b}&{\tt t}&{\tt b}&{\tt t}&{\tt b}\\
{\tt n}&{\tt t}&{\tt t}&{\tt n}&{\tt n}\\
{\tt f}&{\tt t}&{\tt t}&{\tt t}&{\tt t}\\
\end{tabular}
\qquad
\begin{tabular}{c|cccc}
\,$\to$\,&{\tt t}&{\tt b}&{\tt n}&{\tt f}\\
\hline
{\tt t}&{\tt t}&{\tt b}&{\tt n}&{\tt f}\\
{\tt b}&{\tt t}&{\tt b}&{\tt n}&{\tt f}\\
{\tt n}&{\tt t}&{\tt t}&{\tt t}&{\tt t}\\
{\tt f}&{\tt t}&{\tt t}&{\tt t}&{\tt t}\\
\end{tabular}
\qquad
\begin{tabular}{c|cccc}
\,$\hookrightarrow$\,&{\tt t}&{\tt b}&{\tt n}&{\tt f}\\
\hline
{\tt t}&{\tt t}&{\tt b}&{\tt n}&{\tt f}\\
{\tt b}&{\tt t}&{\tt t}&{\tt n}&{\tt n}\\
{\tt n}&{\tt t}&{\tt b}&{\tt t}&{\tt b}\\
{\tt f}&{\tt t}&{\tt t}&{\tt t}&{\tt t}\\
\end{tabular}

\vspace{.5cm}
\begin{tabular}{c|cccc}
\,$\stackrel{*}{\to}$\,&{\tt t}&{\tt b}&{\tt n}&{\tt f}\\
\hline
{\tt t}&{\tt t}&{\tt f}&{\tt n}&{\tt f}\\
{\tt b}&{\tt t}&{\tt b}&{\tt n}&{\tt f}\\
{\tt n}&{\tt t}&{\tt n}&{\tt t}&{\tt n}\\
{\tt f}&{\tt t}&{\tt t}&{\tt t}&{\tt t}\\
\end{tabular}
\qquad
\begin{tabular}{c|cccc}
\,$\stackrel{*}{\hookrightarrow}$\,&{\tt t}&{\tt b}&{\tt n}&{\tt f}\\
\hline
{\tt t}&{\tt t}&{\tt f}&{\tt f}&{\tt f}\\
{\tt b}&{\tt t}&{\tt t}&{\tt f}&{\tt f}\\
{\tt n}&{\tt t}&{\tt f}&{\tt t}&{\tt f}\\
{\tt f}&{\tt t}&{\tt t}&{\tt t}&{\tt t}\\
\end{tabular}
}
\end{center}

\caption{Truth tables of implications}
\label{fig:truth-tables-imp}
\end{figure}

\smallskip
Looking at the truth tables of the two implications $\to$ and $\hookrightarrow$, when the left hand side is valid in $S$, it is necessary that the right hand side be also valid in order to make the implication valid.  More precisely, if $\Phi_1$ is valid, the implications $\Phi_1 \to \Phi_2$ and $\Phi_1 \hookrightarrow \Phi_2$ are valid in $S$ for any truth assignment $v$ such that:

$-$~$v(\Phi_1)={\tt t}$ and $v(\Phi_2)={\tt t}$ or $v(\Phi_2)={\tt b}$,

$-$~$v(\Phi_1)={\tt b}$ and $v(\Phi_2)={\tt t}$ or $v(\Phi_2)={\tt b}$.

\smallskip\noindent
As a consequence,  if it happens that $\Phi_1$  is valid while $\Phi_2$ is not, the implication can be made valid  by changing the truth value of $\Phi_2$ in two  ways: making it either true or inconsistent. As will be seen later, we choose to set $v_S(\Phi_2)$ as equal to $v_S(\Phi_1)$. This choice is motivated by the fact that it is the only one satisfying $v(\Phi_1)\preceq_k v(\Phi_2)$ and $v(\Phi_1)\preceq_t v(\Phi_2)$. 

\begin{figure}[t]
\begin{center}
{\footnotesize
\begin{tabular}{c|c}
\,$\phi$\,&\,${\bf T}\phi$\,\\
\hline
{\tt t}&{\tt t}\\
{\tt b}&{\tt f}\\
{\tt n}&{\tt f}\\
{\tt f}&{\tt f}\\
\end{tabular}
\qquad
\begin{tabular}{c|c}
\,$\phi$\,&\,${\bf B}\phi$\,\\
\hline
{\tt t}&{\tt f}\\
{\tt b}&{\tt t}\\
{\tt n}&{\tt f}\\
{\tt f}&{\tt f}\\
\end{tabular}
\qquad
\begin{tabular}{c|c}
\,$\phi$\,&\,${\bf N}\phi$\,\\
\hline
{\tt t}&{\tt f}\\
{\tt b}&{\tt f}\\
{\tt n}&{\tt t}\\
{\tt f}&{\tt f}\\
\end{tabular}
\qquad
\begin{tabular}{c|c}
\,$\phi$\,&\,${\bf F}\phi$\,\\
\hline
{\tt t}&{\tt f}\\
{\tt b}&{\tt f}\\
{\tt n}&{\tt f}\\
{\tt f}&{\tt t}\\
\end{tabular}
\qquad
\begin{tabular}{c|c}
\,$\phi$\,&\,$\circ \phi$\,\\
\hline
{\tt t}&{\tt f}\\
{\tt b}&{\tt f}\\
{\tt n}&{\tt t}\\
{\tt f}&{\tt t}\\
\end{tabular}
}
\end{center}
\caption{More truth tables}
\label{fig:more-tables}
\end{figure}

To see how to express FDE implication $\to$ in terms of the basic connectors $\neg$, $\vee$, $\wedge$, $\not\sim$, $\oplus$ and $\otimes$ of \cite{Tsoukias}, we recall that $\sim$ is defined  for every formula $\phi$ by:

\smallskip\centerline{
$\sim \phi = \neg \not\sim\neg\not\sim \phi = \not\sim\neg\not\sim \neg \phi$.}

\smallskip\noindent
Moreover,  the additional connectors ${\bf T}$, ${\bf B}$, ${\bf N}$ and ${\bf F}$,  whose truth tables are shown in Figure~\ref{fig:more-tables}, allow to `characterize' each truth value in terms of only the standard ones, namely ${\tt t}$ and ${\tt f}$. Roughly speaking, given a truth value ${\tt v}$, the corresponding connector which we denote by ${\bf V}$, is defined for every formula $\phi$ by the fact that ${\bf V}\phi$ is true if  $\phi$ has the truth value  ${\tt v}$ and false otherwise.

In what follows, {\em equivalent} formulas $\phi_1$ and $\phi_2$ are defined as formulas having the same truth tables, which is denoted by $\phi_1 \equiv \phi_2$. Using this notation, it is shown in \cite{Tsoukias} that for each of these connectors, the following equivalences hold:

\smallskip\centerline{
${\bf T}\phi \equiv \phi \wedge  \sim \neg \phi$ ; ${\bf B}\phi \equiv  \not\sim \phi \wedge  \not\sim \neg \phi$ ; ${\bf N}\phi \equiv  \sim \not\sim \phi \wedge  \neg  \not\sim \phi$ ; ${\bf F}\phi \equiv  \sim \phi \wedge  \neg\phi$.}

\smallskip\noindent
We now consider an additional connector denoted by $\circ$, and defined as follows:

\smallskip\centerline{$\circ\phi = {\bf N(\phi)} \vee {\bf F}(\phi)$.}

\smallskip\noindent
 This new connector  `characterizes' the non validity of a formula $\phi$ in terms of  the truth values ${\tt t}$ and ${\tt f}$. In other words, as shown in Figure~\ref{fig:more-tables}, $\circ \phi$ is true if $\phi$ is not valid and false otherwise. 
%
%
%

An important point is that this new connector allows for an intuitively appealing expression of the FDE implication (\cite{Arieli1998,Hazen17})  $\to$. It is indeed easy to show based on the truth tables of Figure~\ref{fig:truth-tables-imp} and Figure~\ref{fig:more-tables}, that for all formulas $\phi_1$ and $\phi_2$, the following equivalence holds:

\smallskip\centerline{
$\phi_1 \to \phi_2 \equiv \circ \phi_1 \vee \phi_2$.}

\smallskip\noindent
Since $\circ \phi$ can be read as {\em true if $\phi$ is not valid and false otherwise}, the equivalence above suggests that $\phi_1 \to \phi_2$ can be read as {\em either $\phi_1$ is not valid or $\phi_2$ is valid}. We emphasize that this is pretty much like implication in standard FOL that is read as {\em either not $\phi_1$ is true or $\phi_2$ is true}.

Based on these remarks and on truth tables in Figures~\ref{fig:truth-tables-con}--\ref{fig:more-tables}, the following proposition holds. The first item in this proposition is the subject of some comments in the next section.
\begin{proposition}\label{prop:implication}
Given formulas $\phi_1$, $\phi_2$ and $\phi_3$, the following equivalences hold:

\smallskip
$-$~$(\phi_1 \vee \phi_2) \to \phi_3 \equiv   (\phi_1 \oplus \phi_2) \to \phi_3 \equiv   (\phi_1 \to \phi_3)\wedge (\phi_2 \to \phi_3)$

$-$~$(\phi_1 \wedge \phi_2) \to \phi_3 \equiv   (\phi_1 \otimes \phi_2) \to \phi_3 \equiv   (\phi_1 \to \phi_3)\vee (\phi_2 \to \phi_3)$.
\end{proposition}
%
%
\subsection{About Functional Completeness}
Functional completeness in our context can be stated as follows: Given a function $W$ from $\{{\tt t}, {\tt b}, {\tt n}, {\tt f}\}^k$ to $\{{\tt t}, {\tt b}, {\tt n}, {\tt f}\}$ where $k$ is a positive integer, can $W$ be `expressed' as a formula $\Phi_W(P_1, P_2, \ldots , P_k)$ involving $k$ propositional variables $P_1, P_2, \ldots, P_k$? More formally, given $W$, the problem is to prove that there exists a formula $\Phi_W$ such that for ${\tt V}=({\tt v}_1,{\tt v}_2 ,\ldots ,{\tt v}_k)$ in $\{{\tt t}, {\tt b}, {\tt n}, {\tt f}\}^k$, if $v$ is a valuation such that for $i=1, 2, \ldots, k$, $v(P_i)={\tt v}_i$, then $v(\Phi_W({\tt v}_1, {\tt v}_2, \ldots , {\tt v}_k))=W({\tt V})$.

This question has been answered positively in \cite{Arieli1998} where the proposed formula $\Phi_W$ involves the connectors $\neg$, $\wedge$ and $\to$ and the constants ${\tt b}$ and ${\tt n}$. The authors give also some other variants of this result by proposing various sets of connectors,  all of which containing the implication $\to$.

Given that $\phi_1 \to \phi_2$ can be expressed as $\circ \phi_1 \vee \phi_2$, functional completeness can also be shown based on the connectors introduced in \cite{Tsoukias}, that is  $\neg$, $\not\sim$, $\vee$, $\wedge$, $\oplus$ and $\otimes$, but not $\to$. We prove this result in two ways: one based on \cite{Arieli1998}, and one more direct, using the connectors defined in \cite{Tsoukias}.

\smallskip\noindent
{\bf{Proof based on \cite{Arieli1998}.}} 
In \cite{Arieli1998}, it is shown that the language $L^*=\{\neg, \wedge, \to, {\tt n}, {\tt b}\}$  is functionally complete, meaning that for every $k \geq0$ and every function $W$ from $\{{\tt t}, {\tt b}, {\tt n}, {\tt f}\}^k$ to $\{{\tt t}, {\tt b}, {\tt n}, {\tt f}\}$ there exists a formula $\Phi^*_W$ in $L^*$ involving $k$ propositional variables $P_1, P_2, \ldots, P_k$ such that, for ${\tt V}=({\tt v}_1,{\tt v}_2 ,\ldots ,{\tt v}_k)$ in $\{{\tt t}, {\tt b}, {\tt n}, {\tt f}\}^k$, if $v$ is a valuation such that for $i=1, 2, \ldots, k$, $v(P_i)={\tt v}_i$, then $v(\Phi^*_W({\tt v}_1, {\tt v}_2, \ldots , {\tt v}_k))=W({\tt V}).$

Thus, given $W$ from $\{{\tt t}, {\tt b}, {\tt n}, {\tt f}\}^k$ to $\{{\tt t}, {\tt b}, {\tt n}, {\tt f}\}$, by replacing in $\Phi^*_W$ every occurrence of $\phi_1 \to \phi_2$ by $\circ\phi_1 \vee \phi_2$ we obtain a formula $\Phi^\circ_W$ that, using the definitions of $\circ$ and of the connectors {\bf N} and {\bf F}, can be expressed by using the basic connectors $\neg,$ $\wedge,$ $\vee,$ $\oplus,$ $\otimes,$ $\not\sim$ and the four truth values.

%
%
%
%
%
%
\smallskip\noindent
{\bf{Direct proof based on \cite{Tsoukias}.}} 
Based on the connectors ${\bf T}$, ${\bf B}$, ${\bf N}$ and ${\bf F}$ introduced in \cite{Tsoukias}, every  ${\tt V}=({\tt v}_1,{\tt v}_2 ,\ldots ,{\tt v}_k)$ in $\{{\tt t}, {\tt b}, {\tt n}, {\tt f}\}^k$ is associated with a formula $\phi_{{\tt V}}(P_1, P_2, \ldots , P_k)$ defined as follows:

\smallskip\centerline{
$\phi_{{\tt V}}(P_1, P_2, \ldots , P_k) = \bigwedge_{i=1}^{i=k}\phi_i(P_i)$}

\smallskip\noindent
where, for $i=1,2, \ldots ,k$, $\phi_i(P_i)={\bf T}P_i$ if ${\tt v}_i ={\tt t}$, $\phi_i(P_i)={\bf B}P_i$ if ${\tt v}_i ={\tt b}$, $\phi_i(P_i)={\bf N}P_i$ if ${\tt v}_i ={\tt n}$ and $\phi_i(P_i)={\bf F}P_i$ if ${\tt v}_i ={\tt f}$.

It is thus easy to see that $v(\phi_{{\tt V}}(P_1, P_2, \ldots , P_k))={\tt t}$ if for $i=1,2, \ldots ,k$, $v(P_i)={\tt v}_i$ and $v(\phi_{{\tt V}}(P_1, P_2, \ldots , P_k))={\tt f}$ otherwise.


\smallskip
Now, given a function $W$ from $\{{\tt t}, {\tt b}, {\tt n}, {\tt f}\}^k$ to $\{{\tt t}, {\tt b}, {\tt n}, {\tt f}\}$, we consider the partition induced by $W$ on $\{{\tt t}, {\tt b}, {\tt n}, {\tt f}\}^k$, defined by $\{W^{-1}({\tt t}),$ $W^{-1}({\tt b}),$ $W^{-1}({\tt n}),$ $W^{-1}({\tt f})\}$.  For every truth value ${\tt v}$ in $\{{\tt t}, {\tt b}, {\tt n}, {\tt f}\}$, the corresponding element  $W^{-1}({\tt v})$ of this partition, which is a subset of $\{{\tt t}, {\tt b}, {\tt n}, {\tt f}\}^k$, is associated with a formula $\Phi_{\tt v}$ defined by: 

\smallskip\centerline{
$\Phi_{\tt v} = \bigvee_{{\tt V} \in W^{-1}({\tt v})} \phi_{{\tt V}}$.}

\smallskip\noindent
It can be seen that for every ${\tt v}$ in $\{{\tt t}, {\tt b}, {\tt n}, {\tt f}\}$, $v(\Phi_{\tt v}) = {\tt t}$ if $(v(P_1), v(P_2), \ldots , v(P_k))$ is in $W^{-1}({\tt v})$, and $v(\Phi_{\tt v}) = {\tt f}$ otherwise.
The targetted formula $\Phi_W$ is defined by:

\smallskip\centerline{
$\Phi_W=((\Phi_{\tt t} \vee \neg \Phi_{\tt f}) \otimes \sim\not\sim \Phi_{\tt n}) \oplus \not\sim \Phi_{\tt b}$.}

\smallskip\noindent
The proof that $\Phi_W$ is indeed the expected formula is done by successively considering the four possible truth values. For ${\tt V}=({\tt v}_1,{\tt v}_2 ,\ldots ,{\tt v}_k)$, consider the following cases: 
\begin{itemize}
\item
${\tt V} \in W^{-1}({\tt t}):$ In this case, we have that $W({\tt V})={\tt t}$. On the other hand, if $v$ is such that for $i=1,2,\ldots ,k$, $v(P_i)={\tt v}_i$, $v(\Phi_{\tt t}) = {\tt t}$, $v(\Phi_{\tt b}) = {\tt f}$, $v(\Phi_{\tt n}) = {\tt f}$ and $v(\Phi_{\tt f}) = {\tt f}$, $v(\Phi_W)$ evaluates as $v(\Phi_W)=(({\tt t} \vee \neg {\tt f}) \otimes \sim\not\sim {\tt f}) \oplus \not\sim{\tt f}={\tt t}$.
Thus, $W({\tt V})=v(\Phi_W)={\tt t}$.
\item
${\tt V} \in W^{-1}({\tt b}):$  In this case, we have that $W({\tt V})={\tt b}$. On the other hand, if $v$ is such that for $i=1,2,\ldots ,k$, $v(P_i)={\tt v}_i$,   $v(\Phi_{\tt t}) = {\tt f}$, $v(\Phi_{\tt b}) = {\tt t}$, $v(\Phi_{\tt n}) = {\tt f}$ and $v(\Phi_{\tt f}) = {\tt f}$, $v(\Phi_W)$ evaluates as 
$v(\Phi_W)=(({\tt f} \vee \neg {\tt f}) \otimes \sim\not\sim {\tt f}) \oplus \not\sim{\tt t}={\tt b}$.
Thus, $W({\tt V})=v(\Phi_W)={\tt b}$.
\item
${\tt V} \in W^{-1}({\tt n}):$   In this case, we have that $W({\tt V})={\tt n}$. On the other hand, if $v$ is such that for $i=1,2,\ldots ,k$, $v(P_i)={\tt v}_i$, $v(\Phi_{\tt t}) = {\tt f}$, $v(\Phi_{\tt b}) = {\tt f}$, $v(\Phi_{\tt n}) = {\tt t}$ and $v(\Phi_{\tt f}) = {\tt f}$, $v(\Phi_W)$ evaluates as $v(\Phi_W)=(({\tt f} \vee \neg {\tt f}) \otimes \sim\not\sim {\tt t}) \oplus \not\sim{\tt f}={\tt n}$.
Thus, $W({\tt V})=v(\Phi_W)={\tt n}$.
\item
${\tt V} \in W^{-1}({\tt f}):$   In this case, we have that $W({\tt V})={\tt f}$. On the other hand, if $v$ is such that for $i=1,2,\ldots ,k$, $v(P_i)={\tt v}_i$, $v(\Phi_{\tt t}) = {\tt f}$, $v(\Phi_{\tt b}) = {\tt f}$, $v(\Phi_{\tt n}) = {\tt f}$ and $v(\Phi_{\tt f}) = {\tt t}$, $v(\Phi_W)$ evaluates as $v(\Phi_W)=(({\tt f} \vee \neg {\tt t}) \otimes \sim\not\sim {\tt f}) \oplus \not\sim{\tt f}={\tt f}$.
Thus, $W({\tt V})=v(\Phi_W)={\tt f}$.
\end{itemize}
As a consequence, we obtain that $W({\tt V})=v(\Phi_W)$ thus that the formula $\Phi_W$  has the same truth values as the truth values defined by the function $W$.

%
\section{Four-Valued Logic and Databases}\label{sec:database-sem}
\subsection{Database Syntax}
As usual when dealing with deductive databases, the considered alphabet is made of constants, variables and predicate symbols with a fixed arity. We thus assume a fixed set of contants, called {\em universe}  and denoted by  ${\cal U}$. It should be noticed that ${\cal U}$ may be infinite. 

As in traditional approaches, a term $t$ is either a constant from ${\cal U}$ or a variable, an {\em atomic formula} or an atom is a formula of the form $P(t_1, t_2, \ldots , t_k)$ where $P$ is a $k$-ary predicate and for every $i=1,2, \ldots ,k$, $t_i$ is a term. A formula is said to be {\em ground} if it contains no variables. A {\em fact} is a ground atom, that is an atom in which all terms are constants. Moreover, a {\em literal} is either an atom or the negation of an atom. In the former case the literal is said to be {\em positive} and in the latter case it is said to be {\em negative}. The {\em Herbrand Base} associated with  ${\cal U}$ is the set of all facts that can be built up using the constants in ${\cal U}$ and the predicates. Clearly, if ${\cal U}$ is infinite, then so is ${\cal HB}$.

In the traditional two-valued setting under the CWA (Closed World Assumption \cite{Reiter77}), the database extension and the database semantics are sets of facts, meant to be true, and the facts not  in the database semantics are set to be false. In our context of Four-valued logic under the OWA (Open World Assumption),  the database extension and the database semantics may contain facts that are either true, inconsistent or false, assuming that non stored facts are unknown.
To account for this situation, we consider sets of pairs of the form $\langle \varphi, {\tt v}\rangle$ where $\varphi$ is a fact in ${\cal HB}$ and where {\tt v} is one of the values {\tt t}, {\tt b} or {\tt f}, while facts whose truth value is {\tt n} are not stored.
Moreover, such a set $S$ is said to be {\em consistent} if for all distinct pairs $\langle \varphi_1,{\tt v}_1 \rangle$ and $\langle \varphi_2,{\tt v}_2 \rangle$ in $S$, $\varphi_1 \ne \varphi_2$. Consequently a consistent  set $S$ is seen as a valuation $v_S$ defined for every $\varphi$ in ${\cal U}$ by:

\smallskip\centerline{
$v_S(\varphi)={\tt v}$, if $S$ contains a pair  $\langle \varphi, {\tt v}\rangle$ ; $v_S(\varphi)={\tt n}$, otherwise.}

\smallskip\noindent
Consistent sets of pairs are called {\em v-sets}, standing for {\em valuated} sets.
%

Given a v-set $S$ and a ground formula $\Phi$, $\Phi$ is said to be {\em valid in} $S$ if $v_S(\Phi)$ is designated. For example, $P(a) \to Q(b)$ is valid in $S_1=\{\langle P(a), {\tt t}\rangle$, $\langle Q(b), {\tt b}\rangle \}$ because $v_{S_1}(P(a) \to Q(b))={\tt b}$, but $P(a) \to Q(b)$ is not valid in $S_2=\{\langle P(a), {\tt t}\rangle \}$ because $v_{S_2}(P(a) \to Q(b))={\tt n}$.

The two orderings $\preceq_k$ and $\preceq_t$ are extended  to v-sets over the same base ${\cal HB}$ in a point-wise manner as follows.
\begin{definition}\label{def:orderings}
For all v-sets $S_1$ and $S_2$ over ${\cal U}$, $S_1 \preceq_k S_2$, respectively  $S_1 \preceq_t S_2$, holds if for every $\varphi$ in ${\cal U}$, $v_{S_1}(\varphi) \preceq_k v_{S_2}(\varphi)$, respectively $v_{S_1}(\varphi) \preceq_t v_{S_2}(\varphi)$, holds.
\end{definition}
For example for ${\cal HB}=\{P(a),$ $P(b),$ $P(c)\}$, $S_1=\{\langle P(a), {\tt t}\rangle\}$ and $S_2=\{\langle P(a), {\tt b}\rangle ,$ $ \langle P(b), {\tt f}\rangle\}$, we have $v_{S_1}(P(b))=v_{S_1}(P(c))=v_{S_2}(P(c))={\tt n}$. Thus:
\begin{itemize}
\item
$v_{S_1}(P(a)) \preceq_k v_{S_2}(P(a))$,  $v_{S_1}(P(b)) \preceq_k v_{S_2}(P(b))$ and  $v_{S_1}(P(c)) \preceq_k v_{S_2}(P(c))$,  implying that $S_1 \preceq_k S_2$ holds.
\item
$v_{S_2}(P(a)) \preceq_t v_{S_1}(P(a))$,  $v_{S_2}(P(b)) \preceq_t v_{S_1}(P(b))$ and $v_{S_2}(P(c)) \preceq_t v_{S_1}(P(c))$,  implying that $S_2 \preceq_t S_1$ holds.
\item
$\emptyset \preceq_k S_2$, because for every $\varphi$, $v_{\emptyset}(\varphi)={\tt n}$, the least value with respect to $\preceq_k$.
\item
$\emptyset$ and $S_2$ are not comparable with respect to $\preceq_t$, because $v_{\emptyset}(P(a))={\tt n}$ and $v_{S_2}(P(a))={\tt b}$ are not comparable with respect to $\preceq_t$.
\end{itemize}
The extension of $\preceq_k$ generalizes set inclusion in the sense that if $S_1 \subseteq S_2$, then we have $S_1 \preceq_k S_2$. Notice that, as the last item above shows, the truth ordering $\preceq_t$ does not satisfy this property, because $\emptyset \subseteq S_2$ holds while $\emptyset \preceq_t S_2$ does not.

In our context, as in approaches to Datalog databases (\cite{CeriGT90,Bidoit91}), a database consists of an {\em extension}  and a {\em set of rules}, formally defined as follows.
\begin{definition}\label{def:database}
A database $\Delta$ is a pair $\Delta=(E,R)$ where $E$ and $R$ are respectively called the {\em extension} and the {\em rule set} of $\Delta$. If $\Delta=(E,R)$, then:
\begin{itemize}
\item
$E$ is a v-set.
\item
$R$ is a set of rules of the form $\rho:h(X) \leftarrow B(X,Y)$ where the variables in $X$ are free in $h(X)$ and $B(X,Y)$ and the variables in $Y$ are free in $B(X,Y)$, and
\begin{enumerate}
\item
$B(X,Y)$ is a well formed formula involving the connectors $\neg$,  $\vee$, $\wedge$, $\oplus$ and $\otimes$. $B(X,Y)$ is called the {\em body} of $\rho$, denoted by $body(\rho)$.
\item
$h(X)$ is a positive or negative literal, called the {\em head} of $\rho$, denoted by $head(\rho)$.
\end{enumerate}
\end{itemize}
\end{definition}
It should be clear that the rules as defined above generalize standard Datalog$^{neg}$ rules (\cite{Bidoit91}). On the other hand, the definition above also generalizes rules as defined in \cite{Lau2019} where the bodies of the rules are restricted to be conjunctions only. Moreover, in our approach and contrary to \cite{Fitting91,Bidoit91}, rules may generate contradictory facts. It is important to notice that our approach is closely related to the generalized rules as introduced  in \cite{Fitting91}, with the following notable differences:
\begin{enumerate}
\item
In our approach, negative literals are allowed in the rule heads, which is not the case in \cite{Fitting91}.
\item
In our approach, several rules may have the same predicate involved in their head,  which is not the case in \cite{Fitting91}. This important point will be discussed later.
\item
In our approach, quantifiers are not allowed, whereas in \cite{Fitting91} four quantifiers are allowed ($\forall$ and $\exists$ associated with $\preceq_t$ and {\large $\Pi$} and {\large $\Sigma$} associated with $\preceq_k$).
\end{enumerate}
\subsection{Database Semantics}
As usual, rules are seen as implications, either $\to$ or $\hookrightarrow$ that must be valid in the database semantics. Notice in this respect that Figure~\ref{fig:truth-tables-imp} shows that for all formulas $\phi_1$ and $\phi_2$, $\phi_1\to \phi_2$ is valid if and only if so is $\phi_1\hookrightarrow \phi_2$. This explains why in \cite{Lau2019}, our approach has been shown to be `compatible' with either implication. Here, we focus on FDE implication $\to$, thus forgetting the implication  $\hookrightarrow$ of \cite{Tsoukias}.

Similarly to the standard Datalog approach, a model of a database $\Delta=(E,R)$ could be defined as a v-set $M$ containing $E$ and in which all rules in $R$ are valid. However,  such a definition would raise important problems:
\begin {enumerate}
\item
{\em A database might have no model.} To see this, consider $\Delta=(E,R)$ where $R=\{Q(b) \leftarrow P(a)\}$ and where $E=\{\langle P(a), {\tt t}\rangle$, $\langle Q(b), {\tt f}\rangle\}$. Then in any model $M$,  $v_M(P(a) \to Q(b))={\tt f}$ because $M$ must contain the two pairs of $E$.  Notice that this cannot happen in standard Datalog since the storage of false facts is not allowed.
\item
{\em A database might have more than one minimal model, with respect to set inclusion.} This case is illustrated above where $S'_1=\{\langle P(a), {\tt t}\rangle$, $\langle Q(b), {\tt t}\rangle\}$ are $S'_2=\{\langle P(a), {\tt t}\rangle$, $\langle Q(b), {\tt b}\rangle\}$ two minimal v-sets containing $\{\langle P(a), {\tt t}\rangle\}$ in which $Q(b) \leftarrow P(a)$ is valid. This situation does not happen in standard Datalog because the minimal model is known to be unique.
\end{enumerate}
Whereas the second issue raised above will be further investigated later, the first issue is solved in our approach by giving the priority to the database extension over the rules. To do so, we prevent from applying a rule in $R$ when it leads to some conflict with a v-pair in $E$.

In order to implement this policy, given a database $\Delta=(E,R)$  over universe ${\cal U}$, we denote by $inst(E,R)$ the set of all instantiations $\rho$ of rules in $R$ such that $head(\rho)$ does not occur in $E$. Moreover, given a rule $\rho:head(\rho) \leftarrow body(\rho)$ we denote by $\rho^\to$ the formula $body(\rho) \to head(\rho)$. The definition of a model of $\Delta$ then follows.
\begin{definition}\label{def:model}
Let $\Delta=(E,R)$ be a database. A v-set $M$ is a {\em model} of $\Delta$ if the following holds:

\smallskip
$1.$ $E \subseteq M$, i.e., $M$ must contain the database extension, and

$2.$ every $\rho$ of $inst(E,R)$ is valid in $M$, that is, $v_M(\rho^\to)$  is designated.
\end{definition}
To illustrate Definition~\ref{def:model}, consider the following simple examples:
\begin{itemize}
\item
$\Delta=(E,R)$ with $E=\{\langle P(a), {\tt t}\rangle$, $\langle Q(b), {\tt f}\rangle\}$ and $R=\{Q(b) \leftarrow P(a)\}$. $E$ is a model of $\Delta$ as $inst(E,R)= \emptyset$. It is easy to see that $E$ is the only minimal model with respect to set inclusion.
\item
$\Delta=(E,R)$ with $E=\{\langle P(a), {\tt t}\rangle\}$ and $R=\{Q(b) \leftarrow P(a)\}$ $S_1=\{\langle P(a), {\tt t}\rangle$, $\langle Q(b), {\tt t}\rangle\}$ and $S_2=\{\langle P(a), {\tt t}\rangle$, $\langle Q(b), {\tt b}\rangle\}$ are  two models of $\Delta$. Moreover, it can be seen that these two models are minimal with respect to set inclusion.
\end{itemize}
Given a database $\Delta$, an  immediate consequence operator is defined below. It will then be seen that this  allows for computing a particular model of $\Delta$, which  we call the {\em semantics} of $\Delta$.
\begin{definition}\label{def:sem-operator}
Let $\Delta=(E,R)$ be a database. The {\em semantic immediate consequence operator} associated with $\Delta$, denoted by $\Sigma_{\Delta}$, is defined for every v-set $S$ by the following steps:

\smallskip\noindent
$1.$~Define first $\Gamma^E_{\Delta}(S)$ as follows:
$$\begin{array}{rl}
\Gamma^E_{\Delta}(S) =S~\cup &\{\langle h, {\tt t}\rangle~|~(\exists \rho \in inst(E,R))(h=head(\rho) \wedge v_S(body(\rho))= {\sf t})\}\quad ~\\
\cup &\{\langle h, {\tt b}\rangle~|~(\exists \rho \in inst(E,R))(h=head(\rho) \wedge  v_S(body(\rho))= {\sf b})\}\\
\cup &\{\langle h, {\tt f}\rangle~|~(\exists \rho \in inst(E,R))(\neg h=head(\rho) \wedge  v_S(body(\rho))= {\sf t})\}\\
\cup &\{\langle h, {\tt b}\rangle~|~(\exists \rho \in inst(E,R))(\neg h=head(\rho) \wedge  v_S(body(\rho))= {\sf b})\}
\end{array}
$$
$2.$~Then, define $\Sigma_{\Delta}(S)$ by:
%
$\Sigma_{\Delta}(S) =\{\langle \varphi ,{\tt v}_{\oplus}(\varphi)\rangle~|~\varphi {\mbox{ occurs in }}\Gamma^E_{\Delta}(S)\}$, where
%

${\tt v}_{\oplus}(\varphi)=\bigoplus\{{\tt v}~|~\langle \varphi, {\tt v}\rangle \in \Gamma^E_{\Delta}(S)\}$.
\end{definition}
Definition~\ref{def:sem-operator} should be seen as fitting our view on rule semantics based of FDE implication, whose validity has been expressed earlier as {\em $\phi_1 \to \phi_2$ is valid if and only if whenever $\phi_1$ is valid, so is $\phi_2$}. This point of view is similar to that in Datalog databases (where `valid' means `true'), but different from the one in \cite{Fitting91}, where the truth value of the head of the rule is equated to that of the body, {\em whatever the truth value of the body}, even when it is ${\tt f}$. The following lemma shows basic properties of the operator $\Sigma_\Delta$.
\begin{lemma}\label{lem:basic}
For every database $\Delta=(E,R)$, $\Sigma_\Delta$ is monotonic and continuous with respect to $\preceq_k$.
\end{lemma}
\begin{proof}
We first notice that the connectors involved in rule bodies are monotonic, that is, for all formulas $\phi_1$ and $\phi_2$ involving $\neg$, $\vee$, $\wedge$, $\oplus$ or $\otimes$,  if $S_1$ and $S_2$ are two v-sets such that $S_1 \preceq_k S_2$ then $v_{S_1}(\phi_1) \preceq_k v_{S_2}(\phi_2)$ (this can be checked for each operator based on the truth tables in Figure~\ref{fig:truth-tables-con}).

For every $\varphi$ in ${\cal HB}$ and every $i=1,2$, let $D^+_i(\varphi)$ (respectively $D^-_i(\varphi)$) denote the set of all rules $\rho$ in $inst(E,R)$ such that $v_{S_i}(body(\rho))$ is distinguished in $S_i$ and $head(\rho)=\varphi$ (respectively $head(\rho)=\neg \varphi$). Then, $v_{\Sigma_{\Delta}(S_i)}(\varphi)$ can be defined as follows: 
$$
v_{\Sigma_{\Delta}(S_i)}(\varphi)=v_{S_i}(\varphi)\oplus
\bigoplus_{\rho \in D^+_i(\varphi)}v_{S_i}(body(\rho)) \oplus \bigoplus_{\rho \in D^-_i(\varphi)} \neg v_{S_i}(body(\rho))
$$
By monotonicity of $\oplus$ and $\neg$, we obtain that , if $S_1 \preceq_k S_2$, then for every $\varphi$ in ${\cal HB}$, $v_{\Sigma_{\Delta}(S_1)}(\varphi)\preceq_k v_{\Sigma_{\Delta}(S_2)}(\varphi)$, thus entailing the monotonicity of $\Sigma_{\Delta}$ with respect to $\preceq_k$.
The proof that  $\Sigma_{\Delta}$ is continuous with respect to $\preceq_k$, is as in \cite{Fitting91} (see the proof of Theorem~16) and thus omitted here. 
\end{proof}
As a consequence of Lemma~\ref{lem:basic}, given $\Delta=(E,R)$, let $\left(\Sigma^i\right)_{i \geq 0}$  the sequence defined by

\smallskip
$\Sigma^0=E$, and for every $n \geq 1$, $\Sigma^n=\Sigma_\Delta(\Sigma^{n-1})$
 
\smallskip\noindent
has a limit which is the unique least-fixed point of $\Sigma_\Delta$ that is reached for some ordinal at most $\omega$. This limit, denoted by $\Sigma^*_\Delta$, is called the {\em semantics of} $\Delta$ and the valuation $v_{\Sigma^*_\Delta}$ is denoted by $v_\Delta$.

%
\begin{example}\label{ex:sem}
We illustrate the computation of the semantics in the context of our running example, where $\Delta=(E,R)$ is defined by:

\smallskip\noindent
$-$ $E=\{\langle H_1(101), {\tt f}\rangle$, $\langle H_2(101), {\tt f}\rangle$,  $\langle W_1(101), {\tt t}\rangle$,
$\langle H_2(202), {\tt t}\rangle$,  $\langle W_1(202), {\tt f}\rangle$, \\
\indent\indent ~~$\langle W_2(202), {\tt t}\rangle$, $\langle W_1(303), {\tt f}\rangle \}$
\\
$-$ $R=\{\rho_1, \rho_2,\rho_3,\rho_4,\rho_5,\rho_6,\rho_7\}$, where

\begin{tabular}{lll}
$\rho_1: Humid(x) \leftarrow H_1(x) \oplus H_2(x)$&\quad &$\rho_5: Cure(x)  \leftarrow Humid(x)$\\
$\rho_2: White(x) \leftarrow W_1(x) \oplus W_2(x)$& &$\rho_6: \neg Store(x)  \leftarrow \neg White(x)$\\
$\rho_3: Store(x)  \leftarrow \neg Humid(x) \wedge White(x)$& &$\rho_7: New\_test(x)  \leftarrow \neg White(x)$\\
$\rho_4: \neg Store(x)  \leftarrow Humid(x)$& &
\end{tabular}

\smallskip\noindent
We first note that in case, $inst(E,R)=R$ because no predicate occurring in $E$ appears in the heads of the rules of $R$. On the other hand, variables have only three possible instantiations, namely $101$, $202$ and $303$. The computation of  $\Sigma^*_{\Delta}$ is as follows, starting with $\Sigma^0=E$:
\begin{enumerate}
\item
$\Sigma^1=\Sigma_{\Delta}(\Sigma^{0})$. The rule $\rho_1$ generates $\langle Humid(101), {\tt f}\rangle$ and $\langle Humid(202), {\tt t}\rangle$, and  $\rho_2$ generates $\langle White(101), {\tt t}\rangle$, $\langle White(202), {\tt b}\rangle,$ and $\langle White(303), {\tt f}\rangle$.\\
Since  $\Sigma_\Delta(\Sigma^{0}) = \Gamma^E_{\Delta}(\Sigma^{0}) $ we obtain that $\Sigma^1=E \cup \{\langle Humid(101), {\tt f}\rangle$, $\langle Humid(202), {\tt t}\rangle,$ $\langle White(101), {\tt t}\rangle$, $\langle White(202), {\tt b}\rangle,$ $\langle White(303), {\tt f}\rangle\}$.
\item
$\Sigma^2=\Sigma_{\Delta}(\Sigma^{1})$. The computation involves the 5 rules $\rho_3$ \ldots $\rho_7$ as follows:\\
$-$ $\rho_3$ generates $\langle Store(101), {\tt t}\rangle$, because $\neg Humid(101) \wedge White(101)$ has truth value ${\tt t}$. The other instances of $\rho_3$ do not apply because the body is not valid.\\
$-$ $\rho_4$ and $\rho_5$ generate respectively $\langle Store(202), {\tt f}\rangle$ and $\langle Cure(202), {\tt t}\rangle$ because $Humid(202)$ has truth value ${\tt t}$. The other instances of $\rho_4$ and of $\rho_5$ do not apply because the body is not valid.\\
$-$ $\rho_6$ and $\rho_7$ generate respectively  $\langle Store(202), {\tt b}\rangle$ and $\langle New\_test(202), {\tt b}\rangle$ since $White(202)$ has truth value ${\tt b}$, remembering that $\neg {\tt b} ={\tt b}$. \\
$-$ $\rho_6$ and $\rho_7$ generate respectively  $\langle Store(303), {\tt f}\rangle$ and $\langle New\_test(303), {\tt t}\rangle$ since $White(303)$ has truth value ${\tt t}$.\\
$-$ As $\Gamma^E_{\Delta}(\Sigma^{1}) $ contains $\langle Store(202), {\tt b}\rangle$  and $\langle Store(202), {\tt t}\rangle$, the computation of $\Sigma^2$ consists in integrating these v-pairs into $\langle Store(202), {\tt b}\rangle$, remembering that ${\tt b} \oplus {\tt t}={\tt b}$.
We thus obtain that \\
$\Sigma^2=\Sigma^1 \cup\{\langle Store(101), {\tt t}\rangle$, $\langle Store(202), {\tt b}\rangle$, $\langle Store(303), {\tt f}\rangle$, $\langle Cure(202), {\tt t}\rangle$,\\
\rightline{$\langle New\_test(202), {\tt b}\rangle$, $\langle New\_test(303), {\tt t}\rangle \}$.\qquad}
\item
Since no rule applies on $\Sigma^2$ to produce new v-pairs, the computation stops returning $\Sigma^*_{\Delta}=\Sigma^2$.
\end{enumerate}
We draw attention on that  $\Sigma^*_{\Delta}$ is a model of $\Delta$ because $E \subseteq \Sigma^*_{\Delta}$ and all instantiations of the rules in $R$ are valid. For example the instantiation of $x$ in $\rho_4$ and $\rho_6$ by $202$ is valid in $\Sigma^*_{\Delta}$ because:

\smallskip
$-$~$v_{\Delta}(\rho_4^{\to})={\tt b}$, since $v_{\Delta}(Humid(202))={\tt t}$ and\\
\indent\indent $v_{\Delta}(Store(202))= v_{\Delta}(\neg Store(202))={\tt b}$

$-$~$v_{\Delta}(\rho_6^{\to})= {\tt b}$, since $v_{\Delta}(\neg White(202))=v_{\Delta}(\neg Store(202))={\tt b}$.\hfill$\Box$
\end{example}
The following proposition, shows that $\Sigma^*_\Delta$ is a model of $\Delta$. 
\begin{proposition}\label{prop:model}
Given a database $\Delta =(E,R)$, $\Sigma^*_\Delta$ is a  minimal model of $\Delta$, with respect to set inclusion.
\end{proposition}
\begin{proof}We show that $\Sigma^*_\Delta$ is a model of $\Delta$ by contraposition, assuming that  $\Sigma^*_{\Delta} $ is not a model of $\Delta$. First, we have $E \subseteq \Sigma^0$ and then, as $E \preceq_k \Sigma^*_\Delta$ holds by monotonicity and as no instantiated rule can change the truth value of the facts involved in $E$, we have $E \subseteq \Sigma^*_\Delta$. Thus, assuming that  $\Sigma^*_{\Delta} $ is not a model of $\Delta$ implies that at least one rule $\rho$ of $inst(E,R)$  is not valid in $\Sigma^*_\Delta$. In this case, $head(\rho)$ is not valid, while  $body(\rho)$ is valid. Then, denoting $head(\rho)$ by $\varphi$ (respectively $\neg \varphi$), we have $v_{\Delta} (\varphi)={\tt n}$ or $v_{\Delta}(\varphi)={\tt f}$ (respectively $v_{\Delta}(\varphi)={\tt t}$) along with $v_{\Delta} (body(\rho))$ equal to ${\tt t}$ or ${\tt b}$. Consequently $\Sigma_{\Delta} (\Sigma^*_{\Delta} ) \ne \Sigma^*_{\Delta}$, which is not possible by Definition~\ref{def:sem-operator}. This part of the proof is thus complete.

\smallskip
To show the minimality of $\Sigma^*_\Delta$, we show that for every nonempty subset $\sigma$ of $\Sigma^*_\Delta$,  $S=\Sigma^*_\Delta \setminus \sigma$ cannot be a model of $\Delta$. To this end, assuming that $S$ is a model of $\Delta$, let $k$ be the least integer such that $\Sigma^{k-1} \cap \sigma =\emptyset$ and $\Sigma^{k} \cap \sigma \ne \emptyset$. We notice that $k$ exists such that $k>0$ because, since $S$ is a model of $\Delta$, it holds that $E \subseteq S$ and so, since $\Sigma^0=E$, we have  $\Sigma^{0} \cap \sigma =\emptyset$.

Let $\langle \varphi, {\tt v}\rangle$ be in $\Sigma^{k} \cap \sigma$ but not in $\Sigma^{k-1}$. In this case, $v_S(\varphi)={\tt n}$ and as above, there exists one rule $\rho$ in $inst(E,R)$  such that $head(\rho)$ is either $\varphi$ or $\neg\varphi$ and  in $\Sigma^{k-1}$, $head(\rho)$ is not valid, while $body(\rho)$ is valid. Since $\Sigma^{k-1}\subseteq S$,  we have $\Sigma^{k-1}\preceq_k S$ and so,  by monotonicity of the connectors involved in $body(\rho)$, $v_{\Sigma^{k-1}}(body(\rho)) \preceq_k v_S(body(\rho))$. As $body(\rho)$ is valid in $\Sigma^{k-1}$, so is it in $S$. Since $head(\rho)$ is not valid in $S$, $\rho$ is not valid in $S$ either. $S$ being assumed to be a model of $\rho$, we obtain  a contradiction, which completes the proof.
\end{proof}
It has been shown in \cite{Lau2019} that, even with conjunctive rules, $\Sigma^*_\Delta$  is not the only minimal model with respect to set inclusion, nor is it a minimal or a maximal model,  with respect to any of the  orderings $\preceq_k$ and $\preceq_t$. However, we also recall from \cite{Lau2019} that, with conjunctive rules whose heads are positive literals ({\em i.e.,} for Dalatog$^{neg}$ rules) all minimal models with respect to set inclusion share the same false facts and the same valid facts.

At this point, we would like to come back to  Proposition~\ref{prop:implication}, and make an important observation regarding the two closely related notions of implication and rule. We recall that the first item in that proposition is the following: 

\smallskip
$-$~$(\phi_1 \vee \phi_2) \to \phi_3 \equiv   (\phi_1 \oplus \phi_2) \to \phi_3 \equiv   (\phi_1 \to \phi_3)\wedge (\phi_2 \to \phi_3)$.

\smallskip\noindent
Now, consider the three implications as sets of instantiated rules: 

\smallskip
$R_1=\{\varphi \leftarrow \phi_1 \vee \phi_2\}$, $R_2=\{\varphi \leftarrow \phi_1 \oplus \phi_2\}$ and $R_3=\{\varphi \leftarrow \phi_1,\, \varphi \leftarrow \phi_2\}$

\smallskip\noindent
The important observation here is that when computing the corresponding semantics the results are {\em different}. In other words, the three sets of rules lead to different semantics, although the associated implications are {\em equivalent}.


We illustrate this important observation through the following example, in which we also compare our approach with that in \cite{Fitting91}.

\begin{figure}[t]
\begin{center}
{\footnotesize
\begin{tabular}{c|cccc}
$\Delta_1^{{\tt v}_1, {\tt v}_2}$&{\tt t}&{\tt b}&{\tt n}&{\tt f}\\
\hline
{\tt t}&{\tt t}&{\tt t}&{\tt t}&{\tt t}\\
{\tt b}&{\tt t}&{\tt b}&{\tt t}&{\tt b}\\
{\tt n}&{\tt t}&{\tt t}&{\tt n}&{\tt n}\\
{\tt f}&{\tt t}&{\tt b}&{\tt n}&{\tt n}\\
\end{tabular}
\qquad
\begin{tabular}{c|cccc}
$\Delta_2^{{\tt v}_1, {\tt v}_2}$&{\tt t}&{\tt b}&{\tt n}&{\tt f}\\
\hline
{\tt t}&{\tt t}&{\tt b}&{\tt t}&{\tt b}\\
{\tt b}&{\tt b}&{\tt b}&{\tt b}&{\tt b}\\
{\tt n}&{\tt t}&{\tt b}&{\tt n}&{\tt n}\\
{\tt f}&{\tt b}&{\tt b}&{\tt n}&{\tt n}\\
\end{tabular}
\qquad
\begin{tabular}{c|cccc}
$\Delta_3^{{\tt v}_1, {\tt v}_2}$&{\tt t}&{\tt b}&{\tt n}&{\tt f}\\
\hline
{\tt t}&{\tt t}&{\tt b}&{\tt t}&{\tt t}\\
{\tt b}&{\tt b}&{\tt b}&{\tt b}&{\tt b}\\
{\tt n}&{\tt t}&{\tt b}&{\tt n}&{\tt n}\\
{\tt f}&{\tt t}&{\tt b}&{\tt n}&{\tt n}\\
\end{tabular}
}
\caption{Computing the truth values of $S(a)$}
\label{fig:sem}
\end{center}
\end{figure}

\begin{example}\label{ex:implications}
Let $\Delta_1^{{\tt v}_1, {\tt v}_2}=(E^{{\tt v}_1, {\tt v}_2},R_1)$, $\Delta_2^{{\tt v}_1, {\tt v}_2}=(E^{{\tt v}_1, {\tt v}_2},R_2)$, $\Delta_3^{{\tt v}_1, {\tt v}_2}=(E^{{\tt v}_1, {\tt v}_2},R_3)$  be three families of databases  where ${\tt v}_1$ and ${\tt v}_2$ are truth values in $\{{\tt t},{\tt b},{\tt n}, {\tt f}\}$, $E^{{\tt v}_1, {\tt v}_2}$ is either $\{\langle P(a), {\tt v}_1\rangle$, $\langle Q(a), {\tt v}_2\rangle\}$ when ${\tt v}_1 \ne {\tt n}$ and ${\tt v}_1 \ne {\tt n}$, or $\{\langle P(a), {\tt v}_1\rangle\}$ when ${\tt v}_1\ne {\tt n}$ and ${\tt v}_2={\tt n}$, or $\{\langle Q(a), {\tt v}_2\rangle\}$  when ${\tt v}_1={\tt n}$ and ${\tt v}_2 \ne {\tt n}$, or $\emptyset$  when ${\tt v}_1={\tt v}_2={\tt n}$, and

\smallskip
$-$ $R_1 = \{S(a) \leftarrow P(a) \vee Q(a)\}$,

$-$ $R_2 = \{S(a) \leftarrow P(a) \oplus Q(a)\}$,

$-$ $R_3 = \{S(a) \leftarrow P(a) ,~S(a) \leftarrow Q(a)\}$.

\smallskip\noindent
We are thus considering  $3 \times 16=48$ databases whose semantics are defined by $E^{{\tt v}_1, {\tt v}_2} \cup\{\langle S(a), {\tt v}^{12}_i\rangle\}$ where for $i=1,2,3,$ ${\tt v}^{12}_i$ is the truth value obtained by applying $\Sigma_{\Delta_i^{{\tt v}_1, {\tt v}_2}}$ to the rule(s) in $R_i$ and the v-pairs in  $E^{{\tt v}_1, {\tt v}_2}$. The arrays displayed in Figure~\ref{fig:sem} show these truth values based on ${\tt v}_1$ (the rows  of the arrays) and ${\tt v}_2$ (the columns of the arrays). From left to right, the arrays correspond respectively to the three sets of rules $R_1$, $R_2$ and $R_3$.

For example, the value {\tt t} in row `{\tt b}' and column  `{\tt n}' of the array labelled $\Delta_1^{{\tt v}_1, {\tt v}_2}$  in Figure~\ref{fig:sem}, means that $\langle S(a), {\tt t}\rangle$ belongs to the semantics of $\Delta_1^{{\tt b}, {\tt n}}=(\{\langle P(a), {\tt b}\rangle\} , R_1)$, where $Q(a)$ has truth value ${\tt n}$.

It should be stressed that since all these arrays are pairwise distinct, all three sets $R_1$, $R_2$ and $R_3$ produce different semantics in some cases. As examples it can be seen from Figure~\ref{fig:sem} that:

\smallskip
$-$ for ${\tt v}_1={\tt t}$ and ${\tt v}_2={\tt b}$, $S(a)$ is true in $\Delta_1^{{\tt t}, {\tt b}}$ and in $\Delta_3^{{\tt t}, {\tt b}}$, but false in $\Delta_2^{{\tt t}, {\tt b}}$,

$-$ for ${\tt v}_1={\tt b}$ and ${\tt v}_2={\tt n}$, $S(a)$ is true in $\Delta_1^{{\tt b}, {\tt n}}$ and  false in $\Delta_2^{{\tt b}, {\tt n}}$ and in $\Delta_3^{{\tt b}, {\tt n}}$.

\smallskip\noindent
As a consequence, this implies that contrary to standard Datalog approaches, replacing the rule in $R_1$ by the two rules in $R_3$ has an impact on the database semantics in certain cases, although $R_1$ and $R_3$ yield the equivalent formulas as shown in Proposition~\ref{prop:implication}.
Therefore, the claim in \cite{Fitting91} whereby {\em `There is a standard way in Prolog to combine two program clauses for the same relation symbol,
using equality. Similar ideas carry over to languages based on a wide variety of bilattices\ldots'} does not hold in our approach. This also shows that rule based semantics do not always exactly `coincide' with the semantics of  implication. Consequently, the claim above is debatable even in the approach of \cite{Fitting91}, because no comparison is possible, as it makes no sense in \cite{Fitting91} that more than one rule head involves the same predicate. 
%
%

Referring to our running example, the previous statements show that replacing the rules $\rho_4: \neg Store(x) \leftarrow Humid(x)$ and $\rho_6:\neg Store(x) \leftarrow \neg White(x)$ by the rule $\rho_{46}: \neg Store(x) \leftarrow Humid(x) \vee  \neg White(x)$ would lead to different semantics. Indeed, when considering $\rho_4$ and $\rho_6$, the fact that $Humid(202)$ and $White(202)$ have respective truth values ${\tt t}$ and ${\tt b}$,  implies that $Store(202)$ has truth value ${\tt b}$. On the other hand, Figure~\ref{fig:sem} shows that when considering $\rho_{46}$, the same truth values for $Humid(202)$ and $White(202)$ imply that $Store(202)$ has truth value ${\tt f}$. \hfill$\Box$
\end{example}
\subsection{Safe Rules}
An important issue in rule based databases is that a database can have {\em infinite} semantics when ${\cal HB}$ is infinite. This point is indeed problematic because in such cases, answers to some queries can be infinite, which is not acceptable in practice.

As a simple case, consider  $\Delta=(E,R)$ where $E=\{\langle S(a), {\tt t}\rangle\}$ and  $R=\{P(x,y) \leftarrow Q(x,y) \vee S(x)\}$. Based on the truth table of  $\vee$ shown in Figure~\ref{fig:truth-tables-con}, for all $\alpha$ and $\beta$ in ${\cal U}$, $Q(\alpha, \beta) \vee S(\alpha)$ is true if so is $S(\alpha)$. Hence, $\Sigma^*_\Delta = \{\langle S(a), {\tt t}\rangle\} \cup \{\langle P(a,\beta), {\tt t}\rangle~|~\beta \in {\cal U}\}$, which is infinite when ${\cal U}$ is infinite. 

\smallskip
To cope with this difficulty, we define the notion of {\em safe} rules, inspired by the case of Datalog$^{neg}$ databases.
To see how the approaches are related regarding this issue, let ${\cal D} = (\{S(a)\}, \{P(x,y) \leftarrow \neg Q(x,y) \wedge S(x)\})$ be a Datalog$^{neg}$, whose  semantics is $\{S(a)\}\cup\{ P(a,\beta)~|~\beta \in {\cal U}\}$. This result is somehow similar to that for $\Delta$ above, and the rule in ${\cal D}$ is clearly not safe since the variable $y$ in $\neg Q(x,y)$ occurs in no positive literal in the body of the rule.


To formalize and characterize safe rules in our context, we need some preliminaries as detailed next. First, we adapt the notion of {\em active domain} in relational databases \cite{Ullman} to our approach as follows. Given a universe ${\cal U}$, its associated Herbrand base ${\cal HB}$ and a database $\Delta=(E,R)$ over ${\cal HB}$, we call the {\em active domain of $\Delta$}, denoted by ${\cal A}(\Delta)$,  the subset of ${\cal U}$ containing all the constants occurring in $\Delta$. Then the {\em active Herbrand base of $\Delta$}, denoted by ${\cal AB}(\Delta)$ is the set of all facts in ${\cal HB}$ that only  involve constants in ${\cal A}(\Delta)$. Notice that ${\cal A}(\Delta)$ and ${\cal AB}(\Delta)$ are finite sets, even if ${\cal U}$ is infinite, because $E$ and $R$ are assumed to be finite. The notion of safe rule is defined as follows.
\begin{definition}\label{def:safe-rule}
Given a Herbrand base ${\cal HB}$,  a rule  $\rho$ is said to be {\em safe} if for every database $\Delta = (E, \{\rho\})$ where $E$ is an arbitrary finite v-set involving facts in ${\cal HB}$, $\Sigma^*_\Delta$ is a subset of ${\cal AB}(\Delta)$.
\end{definition}
We first notice that, according to Definition~\ref{def:safe-rule}, allowing variables in the head of a rule not occurring in the body would generate non safe rules, and this explains why in Definition~\ref{def:database}, we have restricted all variables occurring in the heads of the rules to also occur in the bodies.
Indeed, let $\rho:P(x,y) \leftarrow Q(x)$ and $\Delta =( \{\langle Q(a), {\tt t}\rangle\}, \{\rho\})$. Then, we have ${\cal AB}(\Delta)=\{P(a,a), Q(a)\}$ and $\Sigma^*_\Delta = \{\langle Q(a), {\tt t}\rangle\} \cup \{\langle P(a, \beta), {\tt t}\rangle ~|~\beta \in {\cal U}\}$, showing that $\rho$ is not safe according to Definition~\ref{def:safe-rule}. Other examples  not relaxing the restriction in Definition~\ref{def:database} are presented next.
\begin{example}\label{ex:safe-rules}
The rule $\rho:P(x) \leftarrow P_1(x) \oplus P_2(x,y)$  is safe, according to Definition~\ref{def:safe-rule}. Indeed, if $inst$ is an instantiation of $x$ and $y$ such that $inst(body(\rho))$ is valid in $E$, then at least one of the instantiated atoms $P_1(\alpha_1)$ or $P_2(\alpha_2,\beta_2)$ is valid in $E$. Hence, these atoms can not generate a v-pair $P(\gamma)$ where $\gamma$ is different than $\alpha_1$ and $\alpha_2$.

Notice that the above reasoning does not hold for $\rho':P'(x,y) \leftarrow P_1(x) \vee P_2(x,y)$ because  for $\Delta=(\{ \langle P_1(a), {\tt t}\rangle\}, \{\rho'\})$, we have

$\Sigma^*_\Delta = \{\langle P_1(a), {\tt t}\rangle\} \cup \{\langle P'(a, \beta), {\tt t}\rangle ~|~\beta \in {\cal U}\}$, \\
showing that $\rho$ is not safe according to Definition~\ref{def:safe-rule}.\hfill$\Box$
\end{example}
In order to syntactically characterize safe rules, we adapt the usual notion of {\em disjunctive normal form} of a formula to the context of Four-valued logic. To this end, we recall from \cite{Fitting91,Tsoukias} the following standard properties of the connectors of the Four-valued logic:

\smallskip
$-$ $\neg(\phi_1 \vee \phi_2) \equiv \neg\phi_1 \wedge \neg \phi_2$ ; $\neg(\phi_1 \wedge \phi_2) \equiv \neg\phi_1 \vee \neg \phi_2$

$-$ $\neg(\phi_1 \oplus \phi_2) \equiv \neg\phi_1 \oplus \neg \phi_2$ ; $\neg(\phi_1 \otimes \phi_2) \equiv \neg\phi_1 \otimes \neg \phi_2$

$-$ Distributivity: for all distinct binary  connectors $\star$ and $\bullet$ in $\{\vee,\wedge, \oplus, \otimes\}$\\
\indent\indent
$\phi_1 \star ( \phi_2 \bullet \phi_3) \equiv (\phi_1 \star \phi_2) \bullet (\phi_1 \star \phi_3)$.

\smallskip\noindent
Using these properties, any quantifier free formula $\Phi$ can be transformed into its equivalent {\em $\vee\oplus$-normal form} according to the following steps:
\begin{enumerate}
\item
{\em $\vee$-transformation:} $\Phi \equiv \Phi_1 \vee \Phi_2 \vee \ldots \vee\Phi_n$ where for every $i$ in $\{1,2, \ldots , n\}$, $\Phi_i$ does not involve the connector $\vee$.
\item
{\em $\oplus$-tranformation:} For every $i$ in $\{1,2, \ldots , n\}$,  $\Phi_i$ is transformed into its equivalent $\oplus$-normal form $\phi_i^1 \oplus \phi_i^2 \oplus \ldots \oplus \phi_i^{p_i}$ where for $j$ in $\{1,2, \ldots , p_i\}$, $\phi_i^j$ does not involve the connector $\oplus$.
\item
{\em $\vee\oplus$-transformation:} Combining these previous two steps, we obtain:\\
$\Phi \equiv (\phi_1^1 \oplus \phi_1^2 \oplus \ldots \oplus \phi_1^{p_1}) \vee (\phi_2^1 \oplus \phi_2^2 \oplus \ldots \oplus \phi_2^{p_2}) \vee \ldots \vee(\phi_n^1 \oplus \phi_n^2 \oplus \ldots \oplus \phi_n^{p_n})$, where for every $i$ in $\{1,2, \ldots , n\}$ and every $j$ in $\{1,2, \ldots , p_i\}$, $\vee$ and $\oplus$ do not occcur in $\phi_i^j$.
\item
{\em $\wedge\otimes$-transformation:} As for every $i$ in $\{1,2, \ldots , n\}$ and every $j$ in $\{1,2, \ldots , p_i\}$, the only connectors occurring in $\phi_i^j$ are $\neg$, $\wedge$ and $\otimes$, the following equivalent form of $\phi_i^j$ can be computed by applying transformations similar to those above:\\
$\phi_i^j \equiv (\lambda_1^1 \otimes \lambda_1^2 \otimes \ldots \otimes \lambda_1^{r_1}) \wedge (\lambda_2^1 \otimes \lambda_2^2 \otimes \ldots \otimes \lambda_2^{r_2}) \wedge \ldots \wedge(\lambda_q^1 \otimes \lambda_q^2 \otimes \ldots \otimes \lambda_q^{r_q})$,\\
where for every $i$ in $\{1,2, \ldots , q\}$ and every $j$ in $\{1,2, \ldots , r_i\}$, $\lambda_i^j$ is a literal, that is of the form $\varphi$ or $\neg\varphi$ where $\varphi$ is in ${\cal HB}$.
\end{enumerate}
Combining these transformations yields a formula equivalent to $\Phi$, called the {\em $\vee\oplus$-normal form of $\Phi$}.
Based on the truth tables of Figure~\ref{fig:truth-tables-con}, given a formula $\Phi$ involving no variable, for every v-set $S$, $\Phi$ is valid in $S$ if and only if there exist $i_0$ in $\{1,2, \ldots , n\}$ and $j_0$ in $\{1,2, \ldots , p_{i_0}\}$ such that $\phi_{i_0}^{j_0}$ is valid in $S$.
Furthermore, assuming that $\phi_{i_0}^{j_0}$ is written as shown in the last item above, $\phi_{i_0}^{j_0}$ is valid in $S$ if and only if every literal $\lambda$ occurring in the $\wedge\otimes$-transformation of  $\phi_{i_0}^{j_0}$ is valid in $S$, that is $v_S(\lambda)$ is ${\tt t}$ or ${\tt b}$ if $\lambda=\varphi$, and $v_S(\lambda)$ is ${\tt f}$ or ${\tt b}$ if $\lambda=\neg\varphi$.

As a consequence, $\Phi$ is valid in $S$ if and only in the $\vee\oplus$-normal of $\Phi$, there exists a $\vee$- and $\oplus$-free sub-formula $\phi_{i_0}^{j_0}$ for which all involved literals are valid in $S$, and thus occur in $S$ with an appropriate truth value. Based on this important remark, the following proposition can be stated.
\begin{proposition}\label{prop:safe}
Let  $\rho:h(X) \leftarrow B(X,Y)$ be a rule such that $B(X,Y)$ is written in its $\vee\oplus$-normal form using the same notation as above. $\rho$ is safe if and only if for every $i$ in $\{1,2, \ldots , n\}$ and every $j$ in $\{1,2, \ldots , p_i\}$, the sub-formula $\phi_i^j$ involves at least all variables in $X$.
\end{proposition}
\begin{proof}
Assume first  that there exist $i_0$ in $\{1,2, \ldots , n\}$ and $j_0$ in $\{1,2, \ldots , p_{i_0}\}$ such that $\phi_{i_0}^{j_0}$  does not involve all variables in $X$. We write $X$ as $X_1X_2$ to mean that the variables in $X_1$ occur in $\phi_{i_0}^{j_0}$ whereas those in $X_2$ do not. Let $inst$ be an instantiation of the variables in $X_1$ and $\Delta=(E, \{\rho\})$ where $E$ is the set of all v-pairs $\langle \varphi, {\tt  b}\rangle$ such that $\varphi$ occurs in $inst(\phi_{i_0}^{j_0})$. Then $inst(\phi_{i_0}^{j_0})$ is valid in $E$ and so, for every extension $inst^*$ of $inst$ to the variables in $X_2$ or in $Y$, $inst^*(body(\rho))$ is valid in $E$. Hence, $inst^*(h(X_1X_2))$ belongs to the semantics of $\Delta$, meaning that $\rho$ is not safe.

Conversely, if  for every $i$ in $\{1,2, \ldots , n\}$ and every $j$ in $\{1,2, \ldots , p_i\}$, the sub-formula $\phi_i^j$ involves at least all variables in $X$, whatever the valid sub-formula $\phi_{i_0}^{j_0}$ in $B(X,Y)$, the instantiation of the variables in $\phi_{i_0}^{j_0}$ assigns a value to every variable in $X$ implying that the fact involved in $inst(h(X))$ is in ${\cal AB}(\Delta)$. Thus, $\rho$ is safe, and the proof is complete.
\end{proof}
\section{Updates}\label{sec:updates}
We first would like to emphasize that our approach to updates follows the same policy as in our previous work on database updating \cite{Lau1997,Lau1998}, whereby priority is given to the latest updates with respect to the current database semantics. This means that updates are {\em always} taken into account and that their effect can not be overridden when computing the semantics. In this approach, such update persistency holds because instantiated rules whose heads involve a fact occurring in $E$, are not applied. This is made possible by restricting instantiated rules to belong to $inst(E,R)$.

\subsection{Standard Update Semantics}
Notice that, contrary to the traditional 2-valued models, in our approach, facts are stored associated with a truth value.  We emphasize in this respect that, in standard 2-valued approaches under CWA, inserting (respectively deleting) $\varphi$ should be understood as {\em take into account that $\varphi$ becomes {\em true} (respectively {\em false}) in the database}. On the other hand, in our Four-valued approach, an update should rather be seen as a {\em change in the truth value} of a given fact. Formally, updates are defined as follows.

\begin{definition}\label{def:update1}
Let $\Delta = (E, R)$ be a  database  and $\nu=\langle \varphi, {\tt v}\rangle$ a v-pair. The result of the {\em update defined by $\nu$ in $\Delta$} is the database $\Delta_\nu=(E_\nu ,R)$ where $E_\nu$ is defined as follows:

\smallskip
$-$ If $\nu=\langle \varphi, {\tt n}\rangle$ then $E_\nu = E \setminus \{\langle \varphi, v_E(\varphi)\rangle\}$  

$-$ Otherwise, $E_\nu=(E \setminus \{\langle \varphi, v_E(\varphi)\rangle\})\cup \{\nu \}.$
\end{definition}
In terms of truth value, an intuitive way to state Definition~\ref{def:update1} is the following:
\begin{itemize}
\item
If ${\tt v}={\tt n}$, the update requires to set the truth value of $\varphi$ to unknown, which amounts to remove from $E$ any v-pair involving $\varphi$, if any. This corresponds to deletions in standard approaches.
\item
Otherwise, if ${\tt v} \ne {\tt n}$, the update consists in replacing the v-pair in $E$ involving $\varphi$, if any, by the v-pair involved in the update, that is $\nu$.
\end{itemize}
\begin{example}\label{ex:updates1}
In the context of our running example, due to $\langle Store(202), {\tt b}\rangle$ in the database semantics, it is likely that the bag has to be tested again. Assuming that in this case the sensors output the following: $\langle H_1(202), {\tt t}\rangle$, $\langle H_2(202), {\tt t}\rangle$ and $\langle W_1(202), {\tt t}\rangle$, these new v-pairs are inserted and the conflicting ones are deleted, thus resulting in the following updated database extension:

\smallskip\noindent
$E'=\{\langle H_1(101), {\tt f}\rangle$, $\langle H_2(101), {\tt f}\rangle$,  $\langle W_1(101), {\tt t}\rangle$,
$\langle H_1(202), {\tt t}\rangle$,  $\langle H_2(202), {\tt t}\rangle$, \\
\indent\indent
$\langle W_1(202), {\tt t}\rangle$, $\langle W_2(202), {\tt t}\rangle$,
$\langle W_1(303), {\tt f}\rangle \}$.
\hfill$\Box$
\end{example}
%
%
\subsection{Other Possible Update Semantics}
In the context of data integration, traditional updates are not always appropriate. Indeed, suppose that   $\langle \varphi, {\tt t}\rangle$ has to be {\em integrated} in a given database $\Delta=(E, R)$ according to the following policy:
\begin{itemize}
\item
If $E$ contains no v-pair involving $\varphi$ ({\em i.e.,} $\varphi$ is unknown in $\Delta$), then the integration of $\langle \varphi, {\tt t}\rangle$ is processed by inserting the v-pair in $E$.
\item
If $E$ contains the v-pair $\langle \varphi, {\tt t}\rangle$, then  the integration of $\langle \varphi, {\tt t}\rangle$ requires no change.
\item
If $E$ contains the v-pair $\langle \varphi, {\tt f}\rangle$, then  the integration of $\langle \varphi, {\tt t}\rangle$ implies that $\varphi$ becomes {\em inconsistent} in $\Delta$, meaning that $\langle \varphi, {\tt t}\rangle$ should be changed to  $\langle \varphi, {\tt b}\rangle$.
\item
If $E$ contains the v-pair $\langle \varphi, {\tt b}\rangle$, then  the integration of $\langle \varphi, {\tt t}\rangle$ implies that $\varphi$ remains {\em inconsistent} in $\Delta$, meaning that no change is required.
\end{itemize}
The last two cases do {\em not} correspond to standard updates, because in the updated database, the truth value of $\varphi$ is not the one specified in the update. In fact, the truth value of $\varphi$ in the updated database is  defined by $v_E(\varphi) \oplus {\tt t}$. Generalizing this remark, we define  {\em integrative updates} as follows.
\begin{definition}\label{def:update2}
Let $\Delta = (E, R)$ be a database,  $\nu=\langle \varphi, {\tt v}\rangle$ a v-pair and $\diamond$ a well formed binary expression involving the connectors $\neg$, $\vee$, $\wedge$, $\oplus$ or $\otimes$. The {\em integrative update} on $\Delta$ defined by $(\nu, \diamond)$ results in the database $\Delta'=(E', R)$ where $E'$ is defined by:

\smallskip
$-$ If $({\tt v} \diamond \, v_E(\varphi )) = {\tt n}$, $E'=E \setminus\{\langle \varphi, v_E(\varphi )\rangle \}$

$-$ Otherwise, $E'=(E\setminus\{\langle \varphi, v_E(\varphi )\rangle \})\cup\{\langle \varphi, ({\tt v} \diamond \, v_E(\varphi ))\rangle\}$
\end{definition}
We illustrate and comment Definition~\ref{def:update2} below.
\begin{enumerate}
\item
As suggested earlier, standard data integration is expressed  by defining $\diamond$ as $\varphi_1\diamond \varphi_2 = \varphi_1 \oplus \varphi_2$.
\item
Considering the connector $\otimes$ instead of $\oplus$ suggests another kind of data integration: instead of cumulating the knowledge as done with $\oplus$, the result of integration can be seen as the `{\em common} knowledge'. For example, when it comes to integrate $\langle \varphi, {\tt t}\rangle$ in the presence of $\langle \varphi, {\tt f}\rangle$, the result is $\langle \varphi, {\tt n}\rangle$, meaning that $\varphi$ becomes unknown.
Moreover, the integration of $\langle \varphi, {\tt t}\rangle$ in the presence of $\langle \varphi, {\tt b}\rangle$, results in keeping the former v-pair while eliminating the latter. This way of integrating can be seen as a mean to eliminate cases of inconsistency.
\item
However, it could not be suitable to eliminate inconsistency, but on the contrary to preserve it. Namely, in the case above, {\em i.e.,} when integrating $\langle \varphi, {\tt t}\rangle$ in the presence of $\langle \varphi, {\tt b}\rangle$, it might be expected that $\langle \varphi, {\tt b}\rangle$ be kept. As shown in the right most table of Figure~\ref{fig:integration}, our approach allows to take this case into account by defining a connector $\odot$ as follows:\\
\centerline{$\varphi_1 \odot \varphi_2= (\varphi_1 \otimes \varphi_2) \oplus (\varphi_1 \otimes \neg \varphi_1) \oplus (\varphi_2 \otimes \neg \varphi_2).$}
\end{enumerate}
It should be emphasized from Definition~\ref{def:update2} that it is unlikely that {\em any} expression $\diamond$ makes sense for defining  an integration policy. We however notice that the last item above shows that some sophisticated expressions might be relevant.
\begin{figure}[t]
\begin{center}
{\footnotesize
\begin{tabular}{c|cccc}
$\varphi_1\otimes \varphi_2$&{\tt t}&{\tt b}&{\tt n}&{\tt f}\\
\hline
{\tt t}&{\tt t}&{\tt t}&{\tt n}&{\tt n}\\
{\tt b}&{\tt t}&{\tt b}&{\tt n}&{\tt f}\\
{\tt n}&{\tt n}&{\tt n}&{\tt n}&{\tt n}\\
{\tt f}&{\tt n}&{\tt f}&{\tt n}&{\tt f}\\
\end{tabular}
~$\bigoplus$
\begin{tabular}{c|c}
$\varphi_1 \otimes \neg\varphi_1$&\\
\hline
{\tt t}&{\tt n}\\
{\tt b}&{\tt b}\\
{\tt n}&{\tt n}\\
{\tt f}&{\tt n}\\
\end{tabular}
~$\bigoplus$
\begin{tabular}{c|c}
$\varphi_2 \otimes \neg\varphi_2$&\\
\hline
{\tt t}&{\tt n}\\
{\tt b}&{\tt b}\\
{\tt n}&{\tt n}\\
{\tt f}&{\tt n}\\
\end{tabular}

\vspace{.5cm}
\begin{tabular}{c|cccc}
$\varphi_1 \odot \varphi_2$&{\tt t}&{\tt b}&{\tt n}&{\tt f}\\
\hline
{\tt t}&{\tt t}&{\tt b}&{\tt n}&{\tt n}\\
{\tt b}&{\tt b}&{\tt b}&{\tt b}&{\tt b}\\
{\tt n}&{\tt n}&{\tt b}&{\tt n}&{\tt n}\\
{\tt f}&{\tt n}&{\tt b}&{\tt n}&{\tt f}\\
\end{tabular}
}
\caption{Computing the truth table of the expression $\odot$}
\label{fig:integration}
\end{center}
\end{figure}
\begin{example}\label{ex:updates2}
In the context of our running example, we assume that the sensor $H_2$ has been replaced with a new one of another type that allows for the additional answer ${\tt b}$ when the degree of humidity has not been determined properly. Notice that this type of output should be distinguished from the absence of answer that is understood as a failure. However, since the sensor is new, its output has to be carefully taken into account. This can be modeled by {\em integrating} the output of the new sensor with the current content of the database ({\em i.e.,} the output from the old sensor).

As explained above this integration can be done in many different ways, some of which being illustrated below, starting form the database extension $E'$ of Example~\ref{ex:updates1}, containing the v-pairs $\langle H_2(101), {\tt f}\rangle$ and  $\langle H_2(202), {\tt t}\rangle$. We also assume that the values returned by the new sensor are:  $\langle H_2(101), {\tt f}\rangle$,  $\langle H_2(202), {\tt b}\rangle$ and $\langle H_2(303), {\tt t}\rangle$. 

Integrating the  new values with the existing ones in the standard way using $\oplus$ would yield: $\langle H_2(101), {\tt f}\rangle$,  $\langle H_2(202), {\tt b}\rangle$ and $\langle H_2(303), {\tt t}\rangle$, meaning that the new values replace the current ones. However, a more conservative way of integrating the new values is to consider the connector $\otimes$ instead of $\oplus$, which would yield the following:  $\langle H_2(101), {\tt f}\rangle$ and  $\langle H_2(202), {\tt t}\rangle$, meaning that the inconsistency returned by the new sensor is not taken into account and that $H_2(303)$ remains unknown.

Although this result could be seen as more `conservative' than the first one in case of disagreement, it might seem counter-intuitive that the inconsistency is not taken into account. Considering the operator $\odot$ would  produce  $\langle H_2(101), {\tt f}\rangle$ and  $\langle H_2(202), {\tt b}\rangle$, meaning that the inconsistency is now taken into account  and that $H_2(303)$ remains unknown.\hfill$\Box$
\end{example}
We argue that integrative updates generalize standard updates, because any standard update can be expressed as an integrative update. Indeed, given a database $\Delta$ and a fact $\varphi$, the following holds:

\smallskip\noindent
$-$ The update defined by $\langle \varphi, {\tt t}\rangle$ is expressed by the integrative update $(\langle \varphi, {\tt t}\rangle, \vee)$.
\\
$-$ The update defined by $\langle \varphi, {\tt b}\rangle$ is expressed by the integrative update $(\langle \varphi, {\tt b}\rangle, \oplus)$.
\\
$-$ The update defined by $\langle \varphi, {\tt n}\rangle$ is expressed by the integrative update $(\langle \varphi, {\tt n}\rangle, \otimes)$.
\\
$-$ The update defined by $\langle \varphi, {\tt f}\rangle$ is expressed by the integrative update $(\langle \varphi, {\tt f}\rangle, \wedge)$.
\section{Related Work}\label{sec:rel-work}
Comparing our approach with all related work in the literature is simply not possible due to the huge amount of papers on these topics that have been published during the past four or five decades... In what follows, we mainly focus on the most related approaches dealing with $(i)$ logic and databases, $(ii)$ inconsistent databases, $(iii)$ multi-valued logic.

\smallskip\noindent
{\bf Logic and Databases.}
We first refer to \cite{CeriGT90,Ullman,MinkerSZ14} for surveys of standard approaches to Datalog databases, while in \cite{Bidoit91} the problem of negation is overviewed in more details. It is important to recall that in all these work, CWA  is assumed, thus leading to difficulties in handling falsity, a problem that does not arise in our framework,  which assumes OWA instead of CWA.

Changing from CWA to OWA is not new \cite{Bergman} and the need has appeared due to the emergence of data integration on the web. This is so because in this framework, when a piece of information has not been retrieved in the answer to a query, this cannot be seen as that this piece of information is {\em false}, but rather that this piece of information has not been searched properly. It is thus more appropriate that this piece of information be assigned the truth value {\em unknown}.

On the other hand, the examples in this paper suggest that when integrating information from several sources, contradictions may occur, thus motivating for the introduction of {\em inconsistent} as a truth value. This point of view has also been considered in \cite{AmoP07} but in a logical framework that differs from ours. Indeed, in \cite{AmoP07}, the underlying four valued logic is not the one in \cite{Belnap}, although the considered implication looks similar to FDE implication. Morevover, in \cite{AmoP07} the authors consider two negations in the context of CWA and propose an alternating strategy for computing the database semantics, inspired from the strategy in \cite{GelderRS91} with well-founded semantics.

The work in \cite{Fitting91} is much closer to our approach than that in \cite{AmoP07} because the underlying logic in \cite{Fitting91} is that in \cite{Belnap}. However, the reader is referred to the previous sections regarding some main differences between the approach in \cite{Fitting91} and ours. Among these differences, we mention the form of the rules and the semantic operator that in \cite{Fitting91} makes rule heads false when so is the body, whereas in our approach, the truth value is not changed. Related work following this policy of head assignment to false can be found in \cite{Grahne,Grahne-19} where, in the context of relational databases, reasoning with four truth values is modeled as reasoning {\em twice} under two truth values: once to deduce true information and once to deduce false information (inconsistency being information obtained in the two ways of reasoning). However, the context of the work in  \cite{Grahne,Grahne-19} differs from ours and that in \cite{Fitting91} because in \cite{Grahne,Grahne-19}, implications express equality-generating or tuple-generating-constraints instead of rules.

It is also important to recall that the issue of deductive database updating was first addressed in \cite{Reiter92}, and then by many other authors among which we cite \cite{Lau1998}, which was the first approach suggesting to store false facts and to give priority to most recent updates. The present work builds upon these basic ideas in a much wider context.

\smallskip\noindent
{\bf Inconsistent Databases.}
Regarding related work on inconsistent databases, we propose a radically different approach. Indeed, the purpose of previous work dealing with contradictions in databases, is either to define and investigate  `repairs' so as to make the database consistent (\cite{Afrati,Gianluigi}), and/or to identify a set of queries whose answer is independent from any contradiction (\cite{Greco}). Instead, we propose an approach in which inconsistent information can be stored or deduced through rules, and our purpose is not to eliminate or avoid contradictions.

Indeed, our semantics allows for handling inconsistent information as such, thus reflecting real world applications in which true, false, inconsistent and unknown information have to be dealt with, as is the case when data integration is involved. In doing so, we follow the position in  \cite{Gabbay}, in that inconsistent information should not be avoided, but  treated as such by taking appropriate actions when necessary. The issue of taking actions lies beyond the scope of this  paper, because our rules cannot express an information such as `{\em If $\varphi$ is inconsistent then $\phi$}'. Indeed in our formalism such a rule would be expressed as $\phi \leftarrow {\bf B}\varphi$, which is not allowed, but which is the subject of our current research.

The approach in \cite{LoyerSS04} addresses the issue of data inconsistency due to data integration according to  a specific scenario. In \cite{LoyerSS04}, the authors consider that the information consists of facts that a central server collects from autonomous sources and then tries to combine, using  rules that follow the syntax and the semantics of \cite{Fitting91}, and a set of {\em hypotheses} $H$, representing the server's own estimates. In this setting, the authors show how to compute what they call the {\em support of $H$}, defined as the maximal part of $H$ that does not contradict the facts in the database semantics. This notion of support has then been shown to provide hypothesis-based semantics for the class of programs defined in \cite{Fitting91}, and in the case of Datalog$^{neg}$ programs, these semantics have been shown to extend well-founded semantics of \cite{GelderRS91} and Kripke Kleen semantics of  \cite{Fitting85a}.

\smallskip\noindent
{\bf Multi-valued Logic.} The Four-valued logic that we consider in this work has been introduced in \cite{Belnap} and then has motivated many research efforts in the community of research in non standard logic. Again, our aim is not to review all these work, and we refer to \cite{OmoriW17} for a nice review of this topic. Here, we focus on those work that are the most closely related to ours and that have already been cited in many places. In \cite{Arieli1998} the issue of the functional completeness has been addressed among others and their result has of course inspired our concern on this issue, related to FDE implication. On the other hand, the bi-lattice structure of this logic has been widely studied in \cite{Fitting91}, where the concept of logic programs in this framework was first introduced. We recall that the semantics of the rules in \cite{Fitting91} is different from ours in that in \cite{Fitting91}, the head is set to false when the body is false, whereas in our approach, the truth value of the head is not changed in this case. We argue in this respect that our approach follows standard approaches in that implications whose body is not valid are valid, implying that truth values of the head have not to be changed.

More recently, in \cite{Tsoukias}, an implication slightly different than FDE implication (that we have formerly denoted by $\hookrightarrow$) has ben proposed, and a strong relationship between this logic and rough set theory has been established. We recall that it has been shown in \cite{Lau2019} that our approach works with this implication as well, although FDE implication has been chosen in the present paper.
\section{Conclusion}\label{sec:conclusion}
In this paper we have introduced a novel approach to deductive databases dealing with contradictory information. We stress again that this work is motivated by the facts that $(i)$ many contradictions occur in the real world and these contradictions must be dealt with as such, and $(ii)$ data integration is a field where such contradictions are common. To cope with this issue we consider a deductive database approach based on the Four-valued logic initially introduced in \cite{Belnap}. Our database semantics follows FDE implication and has been shown slightly different from that of \cite{Fitting91}. We also recall that in this paper, rules whose head is a {\em negative} literal are allowed and we have shown that contradicting rules could be safely taken into account in our context. Another important contribution of this work is to propose a new kind of update that allows to `combine' the expected truth value of a fact with its current truth value in the database. This updating policy is of particular interest when it comes to {\em integrate} new pieces of information in a given database.

Based on the results reported in this paper, we are investigating the following issues. First, as rules can contradict each other (a situation which frequently happens in real life), it is important to characterize the exact situations when these contradictions happen and if so, which actions have to be taken, as suggested in \cite{Gabbay}. We are investigating this important issue by extending the form of the rules to allow in their body additional connectors introduced \cite{Tsoukias} (such as connector ${\bf B}$ recalled in Section~\ref{sec:background}). Another important extension of this work is the investigation of an {\em algebraic} language that would allow for the definition of a {\em generic} framework and the expression of {\em constraints} on data such as functional dependencies or tuple generating dependencies. Last but not least, based on such an algebra, we strongly believe that the Four-valued framework provides an elegant and efficient tool for defining a new query language devoted to {\em data integration} rather than to data querying or updating. The notion of  {\em integrative updates} as defined in Section~\ref{sec:updates}, will be the starting point of this future work.

\nocite{*}

%
\end{document}